\documentclass[smallcondensed, envcountsect, envcountsame, runningheads, nospthms]{svjour3}     % onecolumn (ditto)
\usepackage{ifthen}
\newboolean{spthms}
\makeatletter
\@ifclasswith{svjour3}{nospthms}{\setboolean{spthms}{false}}{\setboolean{spthms}{true}}
\makeatother
\ifspthms\usepackage{iflang}
\usepackage{fancyhdr}
\usepackage{tikz}
\usetikzlibra
    \smartqed  % flush right qed marks, e.g. at end of proof
    \spnewtheorem{assumption}[theorem]{Assumption}{\bfseries}{}
    \spnewtheorem{notation}[theorem]{Notation}{\bfseries}{}
    \spnewtheorem{observation}[theorem]{Observation}{\bfseries}{\it}
    \usepackage{aliascnt}
    \newcommand{\qedhere}{\qed}
\else
    \usepackage{amsthm}
    \usepackage{thmtools}
    \newcommand{\ownthmSpaceAbove}{5pt}
    \newcommand{\ownthmSpaceBelow}{5pt}
    \declaretheoremstyle[
        spaceabove=\ownthmSpaceAbove,
        spacebelow=\ownthmSpaceBelow,
        headpunct=,
        postheadspace=.5em,
    ]{definition}
    \declaretheoremstyle[
        spaceabove=\ownthmSpaceAbove,
        spacebelow=\ownthmSpaceBelow,
        headpunct=,
        postheadspace=.5em,
        bodyfont=\itshape,
    ]{theorem}
    \makeatletter
    \declaretheoremstyle[
        spaceabove=\ownthmSpaceAbove,
        spacebelow=\ownthmSpaceBelow,
        preheadhook=\renewcommand\@upn{},
        headpunct=,
        headfont=\it,
        postheadspace=.5em,
    ]{remark}
    \makeatother
    \declaretheorem[numberwithin=section,style=definition]{definition}
    \declaretheorem[style=definition,sibling=definition]{example}
    \declaretheorem[style=definition,sibling=definition]{notation}
    \declaretheorem[style=definition,sibling=definition]{assumption}
    \declaretheorem[style=theorem,sibling=definition]{theorem}
    \declaretheorem[style=theorem,sibling=definition]{corollary}
    
    \declaretheorem[style=theorem,sibling=definition]{lemma}
    \declaretheorem[style=theorem,sibling=definition]{proposition}
    
    \declaretheorem[style=remark,sibling=definition]{remark}
    
    \makeatletter
    \renewenvironment{proof}[1][\proofname]{\par
    \pushQED{\qed}%
    \normalfont \topsep6\p@\@plus6\p@\relax
    \trivlist
    \item\relax
    {\itshape
    #1}\hspace\labelsep\ignorespaces
    }{%
    \popQED\endtrivlist\@endpefalse
    }
    \makeatother
\fi

\usepackage[english]{babel}
\usepackage[utf8]{inputenc}
\usepackage{amssymb,amsmath}
\usepackage{iflang}
\usepackage{fancyhdr}
\usepackage{tikz}
\usetikzlibrary{fit,arrows,shapes,patterns,automata,calc}
\usetikzlibrary{decorations.pathreplacing}
\usepackage{tikz-cd}
\usepackage{xspace}
\usepackage{mathrsfs} % for \mathscr
\usepackage{mathtools} % for \mathllap
\usepackage{paralist}
\usepackage{rotating}
\usepackage{dsfont}

\usepackage[footnote,nomargin,marginclue,final]{fixme}
\FXRegisterAuthor{ls}{als}{LS}
\FXRegisterAuthor{sm}{asm}{SM}
\FXRegisterAuthor{tw}{atw}{TW}

\usepackage{substr}
\newdimen\tightpaperwidth
\newdimen\tightpaperheight
\tightpaperwidth=\textwidth
\tightpaperheight=\textheight
\IfSubStringInString{\detokenize{trimmed}}{\jobname}{
\usepackage[marginparwidth=2cm,nohead,nofoot,
            headsep = 8mm, % gap between heading and text
            footskip = 1cm,
            text={\textwidth,\textheight},
            paperwidth=\tightpaperwidth,
            paperheight=\tightpaperheight,
            marginparsep=1mm,
            marginparwidth=1.4cm,
            footskip=1cm,
            centering
            ]{geometry}
}{}%

\usepackage{enumerate}

\newcommand{\C}{\ensuremath{\mathcal{C}}\xspace}
\newcommand{\D}{\ensuremath{\mathcal{D}}\xspace}
\newcommand{\Set}{\ensuremath{\mathsf{Set}}\xspace}

\newcommand{\Nom}{\ensuremath{\mathsf{Nom}}\xspace}
\newcommand{\GSet}{\ensuremath{\perms(\V)\text{-}\text{set}}\xspace}

\newcommand{\colim}{\ensuremath{\operatorname{colim}}\xspace}
\newcommand{\Coalg}{\ensuremath{\mathsf{Coalg}}\xspace}
\newcommand{\Coalgf}{\ensuremath{\mathsf{Coalg}_\mathsf{f}}\xspace}

\newcommand{\argument}{\_\!\_}
\newcommand{\Id}{\ensuremath{\textnormal{Id}}\xspace}
\newcommand{\Pot}{\ensuremath{{\mathcal{P}}}\xspace}
\newcommand{\Potf}{\ensuremath{{\mathcal{P}_\textnormal{\sffamily f}}}\xspace}
\newcommand{\Potfs}{\ensuremath{{\mathcal{P}_\textnormal{\sffamily fs}}}\xspace}
\newcommand{\Bag}{\ensuremath{{\mathcal{B}\xspace}}}

\newcommand{\Cl}{\ensuremath{{\mathscr{C}\hspace{-2pt}\ell}}}
\newcommand{\V}{\ensuremath{\mathcal{V}}}
\newcommand{\N}{\ensuremath{\mathbb{N}}}
\newcommand{\Z}{\ensuremath{\mathbb{Z}}}
\renewcommand{\Im}{\op{Im}}
\newcommand{\id}{\ensuremath{\textnormal{id}}\xspace}

\newcommand{\inj}{\ensuremath{\mathsf{in}}\xspace}
\newcommand{\outl}{\ensuremath{\operatorname{\sf outl}}}
\newcommand{\outr}{\ensuremath{\operatorname{\sf outr}}}

\newcommand{\op}[1]{\ensuremath{\operatorname{\sf #1}}\xspace}
\newcommand{\supp}{\op{supp}}

\newcommand{\sym}{\ensuremath{\mathfrak{S}}\xspace}
\newcommand{\perms}{\ensuremath{\sym_\mathsf{f}}\xspace}

\usepackage{pifont}% http://ctan.org/pkg/pifont -- for \ding
\newcommand{\fr}{\mathbin{\#}}
\newcommand{\fix}{\mathop{\mathsf{fix}}\xspace}
\newcommand{\Fix}{\mathop{\mathsf{Fix}}\xspace}
\renewcommand{\l}{\ensuremath{\lambda}\xspace}
\def\hra{\hookrightarrow}
\newcommand{\epito}{\twoheadrightarrow}
\newcommand{\abstr}[1]{{\ensuremath{\left<#1\right>}}}
\newcommand{\Real}{\mathds{R}}

\usepackage{newunicodechar}
\newunicodechar{×}{\ensuremath{\times}}
\newunicodechar{≤}{\ensuremath{\le}}
\newunicodechar{→}{\ensuremath{\to}}
\newunicodechar{⊗}{\ensuremath{\otimes}}
\newunicodechar{〈}{\ensuremath{\langle}}
\newunicodechar{〉}{\ensuremath{\rangle}}
\newunicodechar{∀}{\ensuremath{\forall}}
\newunicodechar{∃}{\ensuremath{\exists}}

\pgfdeclarelayer{bg}    % declare background layer
\pgfdeclarelayer{middle}% declare layer between background and main
\pgfsetlayers{bg,middle,main}  % set the order of the layers (main is the standard layer)

\tikzset{commutative diagrams/diagrams={
    rounded corners,
}}

\tikzset{shiftarr/.style={
        rounded corners,%
        to path={--([#1]\tikztostart.center)
                     -- ([#1]\tikztotarget.center) \tikztonodes
                     -- (\tikztotarget)},
}}

\tikzset{shiftarrDR/.style={
        rounded corners,%
        to path={--([#1]\tikztostart.center)
                     |- ([#1]\tikztotarget.center) \tikztonodes
                     -- (\tikztotarget)},
}}
\tikzset{
    compactcd/.style={
        row sep = 1cm,
        column sep = 1cm,
    }
}

\newcommand{\descto}[3][]{
    \arrow[draw=none]{#2}[description,fill=none,#1]{#3}
}

\newcommand{\takeout}[1]{\empty}

\numberwithin{equation}{section}

\begin{document}

\title{Regular Behaviours with Names\thanks{This work forms part of the DFG-funded project COAX (MI~717/5-1 and SCHR~1118/12-1)}} 
\subtitle{On Rational Fixpoints of Endofunctors on Nominal Sets}


\author{Stefan Milius\and
        Lutz Schröder\and
        \mbox{Thorsten Wißmann}
}

\authorrunning{S.~Milius\and L.~Schröder\and
  T.~Wißmann} % if too long for running head
\dedication{In fond memory of our colleague and mentor Horst Herrlich}
%The title page should include:
%The name(s) of the author(s)
%A concise and informative title
%The affiliation(s) and address(es) of the author(s)
%The e-mail address, telephone and fax numbers of the corresponding author

\institute{Lehrstuhl für Informatik 8 (Theoretische Informatik) \\
  FAU Erlangen-Nürnberg \\
  \email{mail@stefan-milius.eu, lutz.schroeder@fau.de, thorsten.wissmann@fau.de}
%
%           F. Author \at
%              first address \\
%              Tel.: +123-45-678910\\
%              Fax: +123-45-678910\\
%              \email{fauthor@example.com}           %  \\
%%             \emph{Present address:} of F. Author  %  if needed
}

\date{
%Received: date / Accepted: date
Preliminary Version
}

\maketitle

\begin{abstract}
Nominal sets provide a framework to study key notions of syntax and
semantics such as fresh names, variable binding and
$\alpha$-equivalence on a conveniently abstract categorical
level. Coalgebras for endofunctors on nominal sets model, e.g.,
various forms of automata with names as well as infinite terms with
variable binding operators (such as $\lambda$-abstraction). Here, we
first study the behaviour of orbit-finite coalgebras for functors
$\bar F$ on nominal sets that lift some finitary set functor $F$. We
provide sufficient conditions under which the rational fixpoint
of~$\bar F$, i.e.~the collection of all behaviours of orbit-finite
$\bar F$-coalgebras, is the lifting of the rational fixpoint
of~$F$. Second, we describe the rational fixpoint of the quotient
functors: we introduce the notion of a sub-strength of an endofunctor
on nominal sets, and we prove that for a functor $G$ with a
sub-strength the rational fixpoint of each quotient of $G$ is a
canonical quotient of the rational fixpoint of $G$. As applications,
we obtain a concrete description of the rational fixpoint for functors
arising from so-called binding signatures with exponentiation, such as
those arising in coalgebraic models of infinitary $\lambda$-terms and
various flavours of automata.
% Please provide 4 to 6 keywords which can be used for indexing purposes.
\end{abstract}
%\begin{acknowledgements}
%\end{acknowledgements}

%\tableofcontents

\section{Introduction}
\label{intro}

Nominal sets (or sets with atoms) were introduced by Mostowski and
Fraenkel in the 1920s and 1930s as a permutation model for set
theory. They are sets equipped with an action of the group of finite
permutations on a given fixed set $\V$ of atoms (playing the roles of
names or variables in applications). Gabbay and Pitts~\cite{gp99}
coined the term \emph{nominal sets} for such sets, and use them as a
convenient framework for dealing with binding operators, name
abstraction and structural induction. The notion of support of a
nominal set allows one to define the notions of ``free'' and ``bound''
names abstractly (we recall this in Section~\ref{sec:nom}). For
example, in order to deal with variable binding in the \l-calculus one
considers the functor
\[
L_\alpha X= \V + [\V]X + X \times X
\]
on $\Nom$, the category of nominal sets, expressing the type of term constructors (note that the abstraction functor $[\V]X$ is a quotient of $\V \times X$ modulo renaming ``bound'' variables). Gabbay and Pitts proved that the initial algebra for $L_\alpha$ is formed by all \l-terms modulo $\alpha$-equivalence. This implies that in lieu of having to deal syntactically with the subtle issues arising in the presence of free and bound variables in inductive definitions on terms, one can simply use initiality as a definition principle.  

Recently, Kurz et al.~\cite{nomcoalgdata} have characterized the final coalgebra for $L_\alpha$ (and more generally, for functors arising from so-called binding signatures): it is carried by the set of all infinitary \l-terms (i.e.~finite or infinite \l-trees) with finitely many free variables modulo $\alpha$-equivalence. This then allows defining operations on infinitary \l-terms by coinduction, for example substitution and operations that assign to an infinitary \l-term its normal form computations (e.g.~the B\"ohm, Levy-Longo, and Berarducci trees of a given infinitary \l-term). 
 
But while the final coalgebra of a functor $F$ collects the behaviour
of \emph{all} coalgebras, one is often interested only in behaviours
of coalgebras whose carrier admits a \emph{finite} representation; in
the case of nominal sets this means that the carrier is
\emph{orbit-finite}. In general, for a finitary endofunctor $F$ on a
locally finitely presentable category, the behaviour of $F$-coalgebras
with a finitely presentable carrier is captured by the notion of
\emph{rational fixpoint} for $F$
(see~\cite{iterativealgebras,m_linexp}). This fixpoint lies between
the initial algebra and the final coalgebra for $F$; as a coalgebra,
it is characterized as the final \emph{locally finitely presentable}
coalgebra\footnote{A coalgebra is locally finitely presentable (lfp)
  if every state in it generates a finitely presentable
  subcoalgebra. We are aware of the terminological clash with locally
  finitely presentable \emph{categories} but it seems contrived to
  call coalgebras with the mentioned property by any other
  name.}. Examples of rational fixpoints %in the category of sets -- this is WRONG
include the sets of regular languages, of eventually periodic and
rational streams, respectively, and of rational formal
power-series. For a polynomial endofunctor $F_\Sigma$ on sets
associated to the signature $\Sigma$, the rational fixpoint consists
of Elgot's \emph{regular} $\Sigma$-trees~\cite{elgot}, i.e.~those
(finite and infinite) $\Sigma$-trees that have only finitely many
different subtrees (up to isomorphism). Recently, Milius and
Wißmann~\cite{mw15} gave a description of the rational fixpoint of
$L_\alpha$ on $\Nom$; it is formed by all rational \l-trees modulo
$\alpha$-equivalence.

In this paper we extend the latter result to a description of the
rational fixpoint for an axiomatically defined class of functors. This
class (properly) includes all \emph{binding functors}, i.e.~functors
arising from binding signatures, but also the finite power-set functor
and exponentiation by orbit-finite strong nominal sets, and our class
is closed under coproducts, finite products, composition and quotients
of functors. Unsurprisingly, in the special case of a functor for a
binding signature, the rational fixpoint is formed by the rational
trees over the given binding signature modulo
$\alpha$-equivalence. However, the proof of our more general result is
surprisingly non-trivial, and not related to the one given
in~\cite{mw15} for the special case $L_\alpha$. Instead we take a
fresh approach and first consider endofunctors $\bar F$ on $\Nom$ that
are a lifting of some finitary endofunctor $F$ on sets. In general, it
is easy to see that for such functors $\bar F$ the initial algebra is
a lifting of the inital $F$-algebra. However, the final coalgebra does
not lift in general; for example, for the $\Nom$-functor
$LX = \V + \V \times X + X \times X$, the final coalgebra of the
underlying $\Set$-functor consists of all \l-trees, and the final
$L$-coalgebra in $\Nom$ consists of all \l-trees with finitely many
variables. The rational fixpoint of $L$, on the other hand, does, by
our results, lift from $\Set$ to $\Nom$; however, this does not hold
for arbitrary liftings of finitary functors. We introduce the notion
of a \emph{localizable} lifting (Definition~\ref{def:localizable}) and we
prove that the rational fixpoint of a localizable lifting $\bar F$ on
$\Nom$ is a lifting of the rational fixpoint of $F$ on sets
(Theorem~\ref{thm:lift}).

In order to characterize the rational fixpoint of functors that make
use of the nominal structure, like $L_\alpha$, we then turn our
attention to quotients of a functor $G$ on $\Nom$. In fact, we introduce the
notion of a \emph{sub-strength} (Definition~\ref{def:sstr}) of an endofunctor
on $\Nom$, and we prove that whenever $G$ is equipped with a
sub-strength then the rational fixpoint of any quotient of $G$ is a
canonical quotient coalgebra of the rational fixpoint of $G$
(Corollary~\ref{cor:rho-quotient}).

We will then see that the combination of Theorem~\ref{thm:lift} and
Corollary~\ref{cor:rho-quotient} allows us to obtain the desired description
of the rational fixpoint of a functor arising from a binding signature
in combination with exponentiation by an orbit-finite strong nominal
set. As a special case we obtain that $\varrho L_\alpha$ is formed by
those $\alpha$-equivalence classes of \l-trees which contain at least
one rational \l-tree.

The fact that our results cover exponentiation, which occurs
prominently in functors for various automata models, is based on a
construction that identifies exponentiation by an orbit-finite strong
exponent as a quotient of a polynomial functor.

\section{Preliminaries}
\label{prelim}

We summarize the requisite background on permutations, nominal sets,
and rational fixpoints of functors. We assume that readers are
familiar with basic notions of category theory and with algebras and
coalgebras for an endofunctor, but start with a terse review of the
latter.

Recall that a \emph{coalgebra} for an endofunctor $F: \C \to \C$ is a
pair $(C, c)$ consisting of an object $C$ of $\C$ and a morphism
$c: X \to FX$ called the \emph{structure} of the coalgebra. A
coalgebra \emph{homomorphism} from $(C, c)$ to $(D, d)$ is a
$\C$-morphism $f: C \to D$ such that $d \cdot f = Ff \cdot c$. A very
important concept is that of a \emph{final} coalgebra, i.\,e. an
$F$-coalgebra $t: \nu F \to F(\nu F)$ such that for every
$F$-coalgebra $(C,c)$ there exists a unique homomorphism
$c^\dagger: (C,c) \to (\nu F, t)$. Final coalgebras exist under mild
assumptions on $\C$ and $F$, e.g.~whenever $\C$ is locally presentable
and $F$ is accessible~\cite{mp89}.

Intuitively, an $F$-coalgebra $(C,c)$ can be thought of as a dynamic
system with an object $C$ of states and with observations about the
states (e.g.~output, next states etc.) given by $c$. The type of
observations that can be made about a dynamic systems is described by
the functor $F$. We denote by $\Coalg\, F$ the category of
$F$-coalgebras and their homomorphisms. For more intuition and
concrete examples we refer the reader to introductory texts on
coalgebras~\cite{rutten00,coalgebratutorial,adamek05}.

\begin{assumption}
  Throughout the paper, all \Set-functors $F$ are w.l.o.g.~assumed to
  preserve monos~\cite{Barr93}; for convenience of notation, we will
  in fact sometimes assume that $F$ preserves subset inclusions.
\end{assumption}

\subsection{Permutations}\label{sec:perms}

We first need a few basic observations about permutations, in
particular that every permutation on an infinite set $X$ can be
restricted to each finite subset of $X$.

\begin{definition}
  For a (not necessarily finite) permutation $f: X\to X$ and a finite
  subset $W\subseteq X$, we define the \emph{restriction} of $f$ to
  $W$ as
    \begin{equation}
      \label{permutationRestriction}
      f\vert_W(v) = \begin{cases}
        f(v) & v\in W\\
        f^{-n}(v) & n\ge 0\text{ minimal s.t. }f^{-n}(v) \notin f[W].
      \end{cases}
    \end{equation}
\end{definition}
\noindent Intuitively, the second case of $f\vert_W$ searches
backwards along $f$ for some value that is not used by the first case;
this is visualized in Figure~\ref{figPermRestriction}.
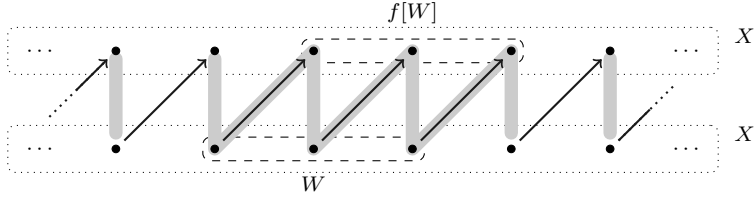
\begin{figure}
\centering
\begin{tikzpicture}[
    x=1.3cm,y=1.3cm,
    element/.style = {
        draw,
        fill,
        shape=circle,
        inner sep = 0cm,
        minimum width=1mm,
    },
    f/.style = {
        draw=black!90,
        thick,
        ->,
        shorten <= 1mm,
        shorten >= 1mm,
    },
    restriction/.style = {
        draw = black!20,
        cap=round,
        line join=round,
%        -{>[scale=2.5,length=5,width=5]},
        thick,
        line width = 5pt,
        shorten <= 6pt,
        shorten >= 3pt,
    },
    set/.style = {
        fill = none,
        draw=black,
        dashed,
        rounded corners= 1mm,
    },
    ]
    % draw the set X as a row twice:
    \coordinate (leftend) at (0, 0);
    \foreach \x [count = \offset ]in {b,...,g} {
        \node[element] (\x top) at (\offset,1) {};
        \node[element] (\x bot) at (\offset,0) {};
    }
    \coordinate (rightend) at (1+\offset, 1);
    % add dots near the ends to indicate infinity
    \node (topLeftDots) [left of=btop] {\(\cdots\)};
    \node (topRightDots) [right of=gtop] {\(\cdots\)};
    \node (botLeftDots) [left of=bbot] {\(\cdots\)};
    \node (botRightDots) [right of=gbot] {\(\cdots\)};
    % draw arrows for f
    \begin{scope}
        \foreach \from/\to in {b/c,c/d,d/e,e/f,f/g} {
            \path[f] (\from bot) -- (\to top);
        }
        % indicate the infinity
        \path[f,shorten <= 11mm] (leftend) -- (btop);
        \path[f,-,shorten >= 11mm] (gbot) -- (rightend) ;
        \path[f,-,shorten <= 6mm,dotted] (leftend) -- (btop);
        \path[f,-,shorten >= 6mm,dotted] (gbot) -- (rightend) ;
    \end{scope}
    % draw sets, i.e. X twice, W and its image f[W]
    % draw the set indicators below everything
    \begin{pgfonlayer}{bg}
        \begin{scope}[
                every label/.style={
                    %draw
                },
                every node/.style={
                    label distance=1mm,
                },
            ]
        \node[fit=(dtop) (ftop),set,label=above:{\footnotesize\(f[W]\)}] {};
        \node[fit=(cbot) (ebot),set,label=below:{\footnotesize\(W\)}] {};
        \node[fit=(topLeftDots) (topRightDots),
              set,dotted,inner sep=4pt,
              label={[yshift=2mm]right:{\footnotesize\(X\)}}] {};
        \node[fit=(botLeftDots) (botRightDots),
              set,dotted,inner sep=4pt,
              label={[yshift=2mm]right:{\footnotesize\(X\)}}] {};
        \end{scope}
    \end{pgfonlayer}
    % draw only the second case of f|w
    \begin{scope}
    \begin{pgfonlayer}{middle}
        \draw[restriction]
                    (fbot.center) -- (ftop.center)
                    -- (ebot.center) -- (etop.center)
                    -- (dbot.center) -- (dtop.center)
                    -- (cbot.center) -- (ctop.center);
        \draw[restriction] (bbot.center) -- (btop.center);
        \draw[restriction] (gbot.center) -- (gtop.center);
    \end{pgfonlayer}
    \end{scope}
\end{tikzpicture}
    \caption{Diagrammatic illustration of \eqref{permutationRestriction}. The
    thin black arrows describe $f$ and the thick grey arrows the second case
    in the definition of $f\vert_W$.}
    \label{figPermRestriction}
\end{figure}
\begin{lemma}
    For any permutation $f: X\to X$ and finite $W\subseteq X$, $f\vert_W$ is
    a finite permutation.
\end{lemma}
\begin{proof}
  We first show that $f|_W(v)$ is indeed defined for all $v$: assume
  that the second case in~\eqref{permutationRestriction} does not
  apply, i.e.~$v \notin W$ with $f^{-n}(v) \in f[W]$ for all $n\ge 0$.
  By finiteness of $f[W]$, we then have $f^{-m}(v) = v$ for some
  $m\ge 1$ and therefore $f(v) = f(f^{-m}(v)) \in f[W]$, which implies
  that $v \in W$, so the first case in~\eqref{permutationRestriction}
  applies.

  For injectivity, let $f\vert_W(u) = f\vert_W(v)$ for $u,v\in \V$ and
  distinguish the following cases:
    \begin{itemize}
    \item For $u,v \in W$, $f(u) =f(v)$ and so $u=v$ as required.
    \item For $u\in W$, $v\notin W$, we have
      $f(u) = f^{-n}(v) \notin f[W]$, contradiction.
    \item For $u,v \notin W$, we have $f^{-n}(u) = f^{-m}(v)$, with $n,m$
    minimal.
    \begin{itemize}
    \item If $n = m$, then $u=v$ as required.
    \item If $n \neq m$, w.l.o.g.~$n > m$, then
      $f(v) = f^{-(n-m-1)}(u) \in f[W]$ by minimality of $n$, since
      $n-m-1 \ge 0$. This implies $v \in W$, contradiction.
    \end{itemize}
  \end{itemize}
  For surjectivity, let $v \in X$.
    \begin{itemize}
    \item If $v\in f[W]$, then $f^{-1}(v) \in W$ and thus $f\vert_W(f^{-1}(v)) =
    v$.
  \item If $v\notin f[W]$, then let $k\ge 0$ be minimal such
    that~$f^{k}(v) \notin W$. Such a $k$ exists because $W$ is finite
    and $v\not \in f[W]$. So we have
    $f\vert_W(f^k(v)) = f^{-k}(f^k(v)) = v$ because firstly
    $f^k(v) \not \in W$, and secondly for all $n < k$,
    \[
        f^{-n}(f^k(v)) = f^{k-n}(v) \in f[W].
        \tag*{\qedhere}
    \]

   \end{itemize}
   This shows that $f|_W$ is a permutation. To see that $f|_W$ is
   finite, note that $f\vert_W(v) = v$ for $v \notin W\cup f[W]$,
   which is a finite set.
 \end{proof}
\begin{remark} \label{permFactor}
    In summary, we have: \begin{inparaenum}
        \item $f\vert_W[W] = f[W]$,
        \item $f\vert_W[f[W]\setminus W] = W \setminus f[W]$, and
        \item $f\vert_W$ fixes every element that is not contained in
          $W$ or $f[W]$.  \end{inparaenum}
    Moreover, the permutation
    \begin{equation*}
      g := f\vert_W^{-1} \cdot f
    \end{equation*}
    maps any $v\in W$ to $g(v) = f\vert_W^{-1}(f(v)) = v$ by
    \eqref{permutationRestriction}. Since $f\vert_W \cdot g = f$, this
    means that we can factor any permutation $f$ into a finite
    permutation $f\vert_W$ and a permutation $g$ that fixes $W$.
\end{remark}
\begin{remark} \label{permFactorComposition} Restriction is compatible
  with composition in the following sense: for permutations $f,h$ and
  finite $W\subseteq X$, we have
    \[
        (f\cdot h)\vert_W(v)
            = f\cdot h(v)
            = f\vert_{h[W]} \cdot h(v)
            = f\vert_{h[W]} \cdot h\vert_{W}(v)
            \quad\text{for all }v\in W
    \]
    using \eqref{permutationRestriction} multiple times. 
\end{remark}

\subsection{Nominal Sets}
\label{sec:nom}

We now briefly recall the key definitions in the theory of nominal sets; see~\cite{Pitts13} for a detailed introduction.

Recall that given a monoid (or more specifically a group) $M$, an
\emph{$M$-set} is a set $X$ equipped with a left action of $M$, which
we denote by mere juxtaposition or by the infix operator
$\cdot$\,. The $M$-sets are the Eilenberg-Moore algebras of the monad
$M\times(-)$\,, which has the unit $\eta(x)=(e,x)$ and
multiplication $\mu(n,(m,x))=(n m,x)$ where $e$ is the unit of
$M$. Given $M$-sets $(X, \cdot)$ and $(Y,\ast)$, a map $f:X\to Y$ is
\emph{equivariant} if $\pi\ast f(x)=f(\pi\cdot x)$ for all
$\pi\in M$, $x\in X$. $M$-sets and equivariant maps form a category,
$M$-set.

We fix a set $\V$ of (variable) \emph{names} (or \emph{atoms}). As
usual, the symmetric group $\sym(\V)$ is the group of all permutations
of $\V$; we denote by $\perms(\V)$ the subgroup of finite permutations
of $\V$, i.e.~the subgroup of $\sym(\V)$ generated by the
transpositions. We have an obvious left action of $\perms(\V)$ on $\V$
given by $\pi\cdot v = \pi (v)$ for any $\pi \in \perms(\V)$ and
$v\in \V$.  Given a $\perms(\V)$-set $X$, we define
\begin{equation*}
  \fix(x)=\{\pi\in \perms(\V)\mid \pi\cdot x=x\}\text{ and }
  \Fix(A)=\{\pi\in \perms(\V)\mid \pi\cdot x=x\text{ for all $x\in A$}\}
\end{equation*}
for $x\in X$ and $A\subseteq X$. We say that a set $A\subseteq\V$ is a
\emph{support} of $x\in X$ or that \emph{$A$ supports $x$} if
\begin{equation*}
  \Fix(A)\subseteq\fix(x),
\end{equation*}
i.e.\ if any permutation that fixes all names in $A$ also fixes
$x$. Moreover, $x\in X$ is \emph{finitely supported} if there exists a
finite set of names that supports $x$. In this case, it can be shown (see e.g.~\cite{Pitts13})
that $x$ has a least support, denoted $\supp(x)$ and called the
\emph{support of $x$}. We say that $v\in\V$ is \emph{fresh} for $x$, and
write $v\fr x$, if $v\in \V \setminus\supp(x)$.

A \emph{nominal set} is a $\perms(\V)$-set $(X,\cdot)$ (or just $X$)
such that all elements of $X$ are finitely supported. We denote by
$\Nom$ the full subcategory of $\GSet$ spanned by the nominal sets. We
have forgetful functors $V:\GSet\to\Set$ and $U:\Nom\to\Set$. Note
that for each nominal set $X$, the function $\supp: X \to \Potf(\V)$
mapping each element to its (finite) support is an equivariant map.

\begin{example}
  \begin{enumerate}[(1)]
  \item The set $\V$ of names with the group action
    $\pi\cdot v = \pi (v)$ is a nominal set; for each $v\in \V$ the
    singleton $\{v\}$ supports $v$.
  \item Every ordinary set $X$ can be made into a nominal set $DX$
    ($D$ for \emph{discrete}) by equipping it with the \emph{trivial
      group action} (also called the \emph{trivial} or \emph{discrete
      nominal structure}) $\pi\cdot x = x$ for all $x\in X$ and
    $\pi \in \perms(\V)$. So each $x\in DX$ has empty support.
  \item The finite \l-terms form a nominal set with the group action
    given by renaming of (free as well as bound!)
    variables~\cite{GabbayPitts99}. The support of a \l-term is the
    set of all variables that occur in it. In contrast, the set of all
    (potentially infinite) \l-trees is not nominal since \l-trees with
    infinitely many variables do not have finite support. However, the
    set of all \l-trees with finitely many variables is nominal.
  \item Given a nominal set $X$, the set $\Potf(X)$ of finite subsets
    of $X$ equipped with the point-wise action of $\perms(\V)$ is a
    nominal set. The support $\supp(Y)$ of $Y\in\Potf(X)$ is the union
    $\bigcup_{x\in Y}\supp(x)$. 
    In particular, the support of each finite $W \in\Potf( \V)$ is $W$
    itself.  Note that $\Pot(\V)$ with the point-wise action is not a
    nominal set because any subset of $\V$ that is neither finite nor
    cofinite fails to be finitely supported. However, the set
    $\Potfs(X)\subseteq \Pot(X)$ of finitely supported subsets of $X$ is a
    nominal set.
  \end{enumerate}
\end{example}

\begin{remark} \label{yetAnotherRemark}
  \begin{enumerate}[(1)]
  \item For an equivariant map $f: X \to Y$ between nominal sets, we
    have $\supp(f(x)) \subseteq \supp(x)$ for any $x\in X$.  To see
    this, let $\pi \in \Fix(\supp(x))$. Then
    \( \pi\cdot f(x) = f(\pi \cdot x) = f(x) \),
    so $\supp(x)$ also supports $f(x)$ and thus
    $\supp(f(x)) \subseteq \supp(x)$.
  \item For $\pi \in \perms(\V)$ and a $\perms(\V)$-set $X$ we denote
    by $\pi_X$ the bijection $X→ X$ defined by
    $x \mapsto \pi \cdot x$. Note that $\pi_X$ fails to be equivariant
    unless $X$ is discrete. However, $\pi$ commutes with all
    equivariant maps $f:X \to Y$ in the sense that
    $f \pi_X = \pi_Y f$; in other words: $\pi: U\to U$ is a natural
    isomorphism.
  \item\label{item:inf-permute} Every nominal set $X$ can be uniquely
    extended to a $\sym(\V)$-set~\cite{gmm06}. By the discussion in
    Section~\ref{sec:perms}, the $\sym(\V)$-action can be defined as
    $\pi\cdot x=\pi|_{\supp(x)}\cdot x$ for $\pi\in\sym(\V)$.  By the
    first item, maps $f$ that are equivariant w.r.t.\ the action of
    $\perms(\V)$ are equivariant also w.r.t.\ the extended action: For
    $\pi\in\sym(\V)$, we have
    $f(\pi\cdot x)=f(\pi|_{\supp(x)}\cdot x)=\pi|_{\supp(x)}\cdot
    f(x)=\pi|_{\supp(f(x))}\cdot f(x)=\pi\cdot f(x)$,
    using in the second-to-last step that $\pi|_{\supp(x)}$ and
    $\pi|_{\supp(f(x))}$ agree on $\supp(f(x))$.
  \end{enumerate}
\end{remark}
\noindent
The category of nominal sets is (equivalent to) a
Grothendieck topos (the so-called Schanuel topos), and so it has rich
categorical structure~\cite{gmm06}. In the following we recall the
structural properties needed in the current paper.

Monomorphisms and epimorphisms in $\Nom$ are precisely the injective and surjective equivariant maps, respectively. It is not difficult to see that every epimorphism in $\Nom$ is strong, i.e., it has the unique diagonalization property w.r.t.~any monomorphism: given an epimorphism $e: A \epito B$, a monomorphism $m: C \hra D$ and $f: A \to C$, $g: B \to D$ such that $g \cdot e = m \cdot f$, there exists a unique diagonal $d: B \to C$ with $d \circ e = f$ and $m \circ d = g$. 

Furthermore, $\Nom$ has image-factorizations; this means that every equivariant map $f: A \to C$ factorizes into an epimorphism $e$ followed by a monomorphism $m$:
    \[
        \begin{tikzcd}[compactcd]
            A \arrow[->>]{r}{e}
        \arrow[shiftarr={yshift=3.5ex} ]{rr}{f}
            &
            B \arrow[hook]{r}{m} &
            C
        \end{tikzcd}
    \]
    Note that the intermediate object $B$ is (isomorphic to) the image
    $f[A]$ in $B$ with the induced action. For an endofunctor $F$ on
    $\Nom$ preserving monos, this factorization system lifts to
    $\Coalg\, F$: every $F$-coalgebra homomorphism $f$ has a
    factorization $f = m \cdot e$ where $e$ and $m$ are $F$-coalgebra
    homomorphisms that are epimorphic and monomorphic in $\Nom$,
    respectively.

Being a Grothendieck topos, $\Nom$ is complete and
cocomplete. Moreover, colimits and finite limits are formed as in
$\Set$, and in fact the forgetful functor $U: \Nom \to\Set$ creates
all colimits and all finite limits~\cite{Pitts03}. Furthermore, $\Nom$
is a locally finitely presentable category~\cite{gu71,ar94}. Recall
that a locally finitely presentable category is a cocomplete category
$\C$ having a set $\mathcal{A}$ of finitely presentable objects such
that every object of $\C$ is a filtered colimit of objects from
$\mathcal{A}$. Petrişan~\cite[Proposition~2.3.7]{petrisanphd} shows
that the finitely presentable objects of $\Nom$ are precisely the
orbit-finite nominal sets:
\begin{definition}
  Given a nominal set $X$ and $x \in X$, the set
  $\{ \pi \cdot x \mid\pi\in \perms(\V)\}$ is called the \emph{orbit}
  of $x$.  A nominal set $(X,\cdot)$ is said to be \emph{orbit-finite} if
  it has only finitely many orbits.
\end{definition}
\noindent The notion of orbit-finiteness plays a central role in our
paper since the rational fixpoint of an endofunctor $F$ on $\Nom$ can
be constructed as the filtered colimit of all $F$-coalgebras with
orbit-finite carrier.

We now collect a few easy properties of orbit-finite sets that we are going to
need. First of all, orbit-finite sets are closed under finite products and
subobjects; hence, under all finite limits (see \cite[Chapter~5]{Pitts13}). And
they are clearly closed under finite coproducts (a well known property of finitely presentable objects) and quotient objects (since the codomain of a surjective equivariant map clearly has fewer orbits); hence under all
finite colimits. 

Generally, for the value of $\pi\cdot x$, it matters only what $\pi$
does on the atoms in $\supp(x)$:
\begin{lemma}\label{onlySuppMatters}
    For $x\in (X,\cdot)$ and any $\pi,\sigma \in\perms(\V)$ with $\pi(v) =
    \sigma(v)$ for all $v\in \supp(x)$, we have $\pi\cdot x = \sigma\cdot x$.
\end{lemma}
\begin{proof}
  Under the given assumptions,
  $\pi^{-1}\sigma\in\Fix(\supp(x))\subseteq\fix(x)$.
\end{proof}

\begin{lemma}\label{orbitsamesupp}
    For any $x_1,x_2 \in X$ in the same orbit, we have $|\supp(x_1)| = |\supp(x_2)|$.
\end{lemma}
\begin{proof}
  By equivariance of $\supp$, $\pi$ induces a bijection between
  $\supp(x)$ and ${\supp(\pi\cdot x)}=\pi\cdot\supp(x)=\pi[\supp(x)]$.
\end{proof}
\begin{lemma}\label{sameorbitfaculty}
    For an element $x$ of a nominal set $X$, there are at most $|\supp(x)|!$
    many elements with support $\supp(x)$ in the orbit of $x$.
\end{lemma}
\begin{proof}
  Let $\pi\in\perms(\V)$ such that $\supp(x) = \supp(\pi \cdot
  x)$. Then
  \[
    \pi[\supp(x)] = \pi \cdot \supp(x) = \supp(\pi\cdot x) = \supp(x), 
  \]
  which shows that $\pi$ restricts to a permutation of $\supp(x)$. 

  If $\sigma$ is also such that $\supp(x) = \supp(\sigma \cdot x)$ and restricts to the same permutation on $\supp(x)$ as $\pi$, then $\pi\cdot x = \sigma \cdot x$ by 
  Lemma~\ref{onlySuppMatters}. Therefore the number of elements in question is at most the
  number of permutations of $\supp(x)$, i.e.~at most $|\supp(x)|!$.
\end{proof}
\noindent One of the properties that make nominal sets interesting for
applications in computer science is that one can think of an element
$x$ of a nominal set as an abstract term and of $\supp(x)$ as the set
of free variables of $x$. It is then possible to speak about
$\alpha$-equivalence on a nominal set, and this leads to Gabbay and
Pitts' \emph{abstraction functor}~\cite[Lemma~5.1]{gp99}:
\begin{definition}
  \label{def:abstr}
  Let $X$ be a nominal set. We define $\alpha$-equivalence
  $\sim_\alpha$ as the relation on $\V×X$ defined by
    \[
    (v_1, x_1) \sim_\alpha (v_2, x_2) \ %:\Leftrightarrow
    \text{ if } (v_1\, z) x_1 = (v_2\, z) x_2 \text{ for }z\fr
    \{v_1,v_2, x_1, x_2\},
    \]
    where the definition of $z \fr M$ spelled out for the case of a finite set
    $M$ means that $z$ is fresh for every element of $M$.
    The $\sim_\alpha$-equivalence class of $(v,x)$ is denoted by
    $\abstr{v}x$.  The abstraction $[\V]X$ of $X$ is the quotient
    $(\V × X)/\mathord{\sim_\alpha}$ with the group action defined by
    \[
        \pi\cdot \abstr{v}x = \abstr{\pi(v)}(\pi\cdot x).
    \]
    For an equivariant map $f: X\to Y$, $[\V]f: [\V]X \to [\V]Y$ is defined by
    $\abstr{v}x \mapsto \abstr{v}(f(x))$.
\end{definition}

\subsection{The Rational Fixpoint}

\noindent Recall that by Lambek's Lemma~\cite{lambek}, the structure
maps of the initial algebra and the final coalgebra for a functor $F$
are isomorphisms, so both yield \emph{fixpoints} of $F$. Here we shall
be interested in a third fixpoint that lies between the initial
algebra and the final coalgebra, the \emph{rational fixpoint} of
$F$. The rational fixpoint can be characterized either as the initial
iterative algebra for $F$~\cite{iterativealgebras} or as the final
locally finitely presentable coalgebra for $F$~\cite{m_linexp}. We
will need only the latter description here.

The rational fixpoint can be defined for any \emph{finitary}
endofunctor $F$ on a locally finitely presentable category $\C$,
i.e.~$F$ is an endofunctor on $\C$ that preserves filtered
colimits. Examples of locally finitely presentable categories are
$\Set$, the categories of posets and of graphs, every finitary variety
of algebras (such as groups, rings, and vector spaces) and every
Grothendieck topos (such as $\Nom$). The finitely presentable objects
in these categories are: all finite sets, posets or graphs, algebras
presented by finitely many generators and relations, and, as we
mentioned before, the orbit-finite nominal sets.

Now let $F: \C \to \C$ be finitary on the locally finitely presentable
category $\C$ and consider the full subcategory $\Coalgf\, F$ of
$\Coalg\, F$ given by all $F$-coalgebras with finitely presentable
carrier.  The \emph{locally finitely presentable} $F$-coalgebras are
characterized as precisely those coalgebras that arise as a colimit of
a filtered diagram of coalgebras from $\Coalgf\, F$~\cite{m_linexp}.
It follows that the \emph{final} locally finitely presentable
coalgebra can be constructed as the colimit of \emph{all} coalgebras
from $\Coalgf\, F$. More precisely, one defines a coalgebra
$r: \varrho F \to F(\varrho F)$ as the colimit of the inclusion
functor of $\Coalgf\, F$:
$ (\varrho F, r) := \colim (\Coalgf\, F \hra \Coalg\, F).  $ Note that
since the forgetful functor $\Coalg\, F \to \C$ creates all colimits,
this colimit is actually formed on the level of $\C$.  The colimit
$\varrho F$ then carries a uniquely determined coalgebra structure $r$
making it the colimit above.

As shown in~\cite{iterativealgebras}, $\varrho F$ is a fixpoint for $F$, i.e.~its coalgebra structure $r$ is an isomorphism. From~\cite{m_linexp} we obtain that local finite presentability of a coalgebra $(C,c)$ has the following concrete characterizations: (1)~for $\C = \Set$ local finiteness, i.e.~every element of $C$ is contained in a finite subcoalgebra of $C$; (2)~for $\C = \Nom$, local orbit-finiteness, i.e.~every element of $C$ is contained in an orbit-finite subcoalgebra of $C$; (3)~for $\C$ the category of vector spaces over a field $K$, local finite dimensionality, i.e., every element of $C$ is contained in a subcoalgebra of $C$ carried by a finite dimensional subspace of $C$.
\begin{example}
  \label{ex:coalgs}
  We list a few examples of rational fixpoints; for more see~\cite{iterativealgebras,m_linexp,bms13}. 
  \begin{enumerate}[(1)]
  \item Consider the functor $FX = 2 \times X^A$ on $\Set$ where $A$ is an input alphabet and $2 = \{0,1\}$. The $F$-coalgebras are precisely the deterministic automata over~$A$ (without initial states). The final coalgebra is carried by the set $\Pot(A^*)$ of all formal languages, and the rational fixpoint is its subcoalgebra of regular languages over $A$.
  \item For $FX = \Real \times X$ on $\Set$, the final coalgebra is carried by the set $\Real^\omega$ of all real streams, and the rational fixpoint is its subcoalgebra of all eventually periodic streams, i.e.~streams $uvvv\cdots$ with $u, v \in \Real^*$. Taking the same functor on the category of real vector spaces, we obtain the same final coalgebra $\Real^\omega$ with the componentwise vector space structure, but this time the rational fixpoint is formed by all rational streams (see~\cite{rutten_rat,m_linexp}).
  \item Recall that in general algebra a finitary signature $\Sigma$ of operation symbols with prescribed arity is a sequence $(\Sigma_n)_{n < \omega}$ of sets. This give rise to an associated polynomial endofunctor $F_\Sigma$ on $\Set$ given by $F_\Sigma X = \coprod_{n < \omega} \Sigma_n \times X^n$. Its initial algebra is formed by all $\Sigma$-terms and its final coalgebra by all (finite and infinite) $\Sigma$-trees, i.e.~rooted and ordered trees such that every node with $n$ children is labelled by an $n$-ary operation symbol. And the rational fixpoint consists precisely of all \emph{rational} $\Sigma$-trees~\cite{elgot,courcelle}, i.e.~those $\Sigma$-trees that have only finitely many different subtrees up to isomorphism~\cite{ginali}.
  \item For the finite powerset functor $\Potf$, the initial algebra is the $\omega$-th step of the cumulative hierarchy of sets, i.e.~$\bigcup_{n < \omega} \Potf^n(\emptyset)$. An isomorphic description is as the set of all finite extensional trees, where a tree is called \emph{extensional} if distinct children of any vertex define non-isomorphic subtrees. A final $\Potf$-coalgebra is carried by the set of all strongly-extensional finitely branching trees, where a tree $t$ is called \emph{strongly-extensional} if for any node $x$ of $t$ no two subtrees of $t$ rooted at $x$ are \emph{tree-bisimilar}; for further explanation and details see~\cite{w05} or~\cite[Corollary~3.19]{almms14}. And the rational fixpoint of $\Potf$ is given by all rational strongly-extensional trees.   
  \item The bag functor $\Bag: \Set \to \Set$ assigns to every set $X$
    the set of all finite multisets on $X$, i.e.~the free commutative
    monoid over $X$. Here we consider trees where children of a vertex are not ordered (in constrast to $\Sigma$-trees in item~(3)), i.e.~the usual graph theoretic notion of tree. Then the initial algebra for $\Bag$ is given by all finite trees, the final coalgebra by all finitely branching ones and the rational fixpoint by all rational ones. This follows from the results in~\cite{am06}.
  \end{enumerate}
\end{example}
\noindent Note that in all the above examples, the rational fixpoint
$\varrho F$ is a subcoalgebra of the final coalgebra
$\nu F$. This need not be the case in general
(see~\cite[Example~3.15]{bms13} for a counterexample). However, we
do have the following result:
\begin{proposition} {\upshape \cite[Proposition 3.12]{bms13}}
  \label{prop:image}
  Suppose that in $\C$, finitely presentable objects are closed under strong
  quotients and that $F$ is finitary and preserves monomorphisms. Then the
  rational fixpoint $\varrho F$ is the subcoalgebra of $\nu F$ given by the
  union of the images of all coalgebra homomorphisms $c^\dagger: (C,c) \to (\nu
  F, t)$ where $(C,c)$ ranges over $\Coalgf\, F$.\footnote{In a general locally
  finitely presentable category the \emph{image} of $c^\dagger$ is obtained by
  taking a (strong epi,mono)-factorization of $c^\dagger$, and the \emph{union} is then  obtained as a directed colimit of the resulting subobjects of $(\nu F, t)$.}
\end{proposition}
\noindent In particular, for a finitary functor $F$ on $\Set$ or
$\Nom$, respectively, that preserves monomorphisms, the rational
fixpoint is the union of the images in $\nu F$ of all finite (or
orbit-finite resp.)  coalgebras; in symbols:
\[
\varrho F = \hspace{-10pt}\bigcup\limits_{\text{$(C,c)$ in $\Coalgf\, F$}} \hspace{-10pt} c^\dagger[C] \ \subseteq\ \nu F.
\]
Note that it is sufficient to let $(C,c)$ range over those coalgebras
in $\Coalgf\, F$ where $c^\dagger$ is injective, because for an
arbitrary (orbit-)finite $(C,c)$ in $\Coalgf F$, its image
$c^\dagger[C]$ is again an (orbit-)finite $F$-coalgebra, and has an
injective structure map.

\section{Liftings of Finitary Functors}
\label{sec:lifting}

We now direct our attention to liftings of finitary \Set-functors, with
a view to investigating their rational fixpoints. We fix some
terminology:
\begin{definition}
  A \emph{lifting} (or, for distinction, a \emph{$\Nom$-lifting}) of a
  functor $F: \Set\to \Set$ is a functor $\bar F: \Nom\to\Nom$ such
  that $U\bar F = F U$. Further, an \emph{$\perms(\V)$-set lifting}
  of~$F$ is a functor $\hat F:\GSet\to\GSet$ such that $V\hat F=FV$.
  We say that a functor $G:\Nom\to\Nom$ is a \emph{lifting} if it is a
  $\Nom$-lifting of some $\Set$-functor $F$.
\end{definition}
\begin{notation}
  Throughout this work, we will use the bar notation in the above
  definition to denote functors on $\Nom$ that are liftings of
  \Set-endofunctors. 
\end{notation}
\begin{lemma}
  A functor $G:\Nom\to\Nom$ is a lifting iff $G$ is a lifting of
  $UGD$, and in fact if $G$ is a lifting of $F$ then $F=UGD$.
\end{lemma}
\begin{proof}
  In the first claim, `if' is trivial; we prove `only if' in
  conjunction with the second claim. So let $UG=FU$ for some \Set-functor $F$;
  then $UGD=FUD=F$.
\end{proof}
\begin{definition}
  Let $\hat F$ be a $\perms(\V)$-set lifting of the functor
  $F: \Set \to \Set$. We say that $\hat F$ is
  \emph{$\Nom$-restricting} if it preserves nominal
  sets, i.e.~$\hat F$ restricts to a functor $\bar F: \Nom\to\Nom$
  (which is, then, a $\Nom$-lifting of $F$). 
\end{definition}
\noindent Recall that a \emph{monad-over-functor distributive law}
between a monad $T$ and a functor $F$ on a category $\C$ is a natural
transformation $\lambda:TF\to FT$ such that the diagrams
\begin{equation}\label{diag:distrib-unit}
    \begin{tikzcd}
      & F\arrow{ld}[above]{\eta F\quad} \arrow{dr}{F\eta}
      \\
      TF \arrow{rr}[below]{\lambda} && FT 
    \end{tikzcd}
  \end{equation}
and
\begin{equation}\label{diag:distrib-mult}
  \begin{tikzcd}
    T^2F \arrow{r}{T\lambda} \arrow{d}[left]{\mu F} & TFT
    \arrow{r}{\lambda T} &
    FT^2 \arrow{d}{F\mu}  \\
    TF \arrow{rr}[below]{\lambda} & & FT
  \end{tikzcd}
  \end{equation}
  commute. Such distributive laws are in bijective correspondence with
  liftings of $F$ to the Eilenberg-Moore category of
  $T$~\cite{johnstone75}. In one direction of this correspondence, we
  obtain from a distribute law $\lambda$ the lifting $\bar F$ that
  maps an Eilenberg-Moore algebra
  \begin{math}
    \begin{tikzcd}
      TA \arrow{r}{a} & A
    \end{tikzcd}
  \end{math}
  to the algebra
  \begin{equation*}
    \begin{tikzcd}
      TFA \arrow{r}{\lambda_A} & FTA \arrow{r}{Fa} & A.
    \end{tikzcd}
  \end{equation*}
  In particular, $\GSet$-liftings of a \Set-functor $F$ are in
  bijection with distributive laws
  \begin{equation}\label{eq:perms-distrib}
    \lambda:\perms(\V)\times F\to F(\perms(\V)\times(-)\,)
  \end{equation}
  of the monad $\perms(\V)\times(-)$\, over $F$.  We will be
  interested exclusively in $\Nom$-liftings that arise by restricting
  $\GSet$-liftings, i.e.~come from a distributive law
  (\ref{eq:perms-distrib}); explicitly:
  \begin{definition}
    A distributive law $\lambda$ of $\perms(\V)\times(-)$\, over
    $F$ is \emph{$\Nom$-restricting} if the corresponding
    $\perms(\V)$-set lifting of $F$ is $\Nom$-restricting.
  \end{definition}
  
  One class of $\Nom$-restricting liftings are \emph{canonical}
  liftings, introduced next. After that, we introduce the bigger class
  of \emph{localizable} liftings, and in the next section we study
  their rational fixpoints.  

  Recall that every \Set-functor
  $F$ comes with a \emph{(tensorial) strength}, i.e.~a transformation
  \begin{equation*}
    s_{X,Y}:X \times FY \to F(X \times Y)
  \end{equation*}
  natural in $X$ and $Y$, making the diagrams
  \begin{equation}
    \label{diag:strength1}
    \begin{tikzcd}
      1\times FY \arrow{dr}[below left]{\iota_{FY}}
                 \arrow{r}{s_{1,Y}}
      & F(1\times Y)
                 \arrow{d}{F\iota_Y}
     \\
      &FY
    \end{tikzcd}
  \end{equation}
  and
  \begin{equation}
    \label{diag:strength2}
    \begin{tikzcd}
      (X\times Z)\times FY \arrow{r}{s_{X\times
          Z,Y}}\arrow{d}[left]{\alpha_{X,Z,FY}} & F((X\times Z)\times Y)
        \arrow{dr}{F\alpha_{X,Z,Y}}\\
        X\times (Z\times FY) \arrow{r}{X\times s_{Z, Y}} & X\times
        F(Z\times Y) \arrow{r}{s_{X,Z\times Y}} & F(X\times(Z\times
        Y))
    \end{tikzcd}
  \end{equation}
  commute, where $\iota:1\times\Id\to\Id$ and
  $\alpha:(\Id\times\Id)\times\Id\to\Id\times(\Id\times\Id)$ are the
  left unitor and the associator, respectively, of the Cartesian
  monoidal structure~\cite{Kock72}.  We remark that using $\iota$ and
  $\alpha$, we can rephrase the definition of the monad $M×(-)$
  for a monoid $(M,m,e)$: the unit is
  \[
    \eta_X \equiv \big(
    \begin{tikzcd}
        X
        \arrow{r}{\iota_X^{-1}}
        & 1×X
        \arrow{r}{e×X}
        & M×X
    \end{tikzcd}
    \big)
  \]
  considering the unit element as a morphism $e: 1\to M$, and the
  multiplication is
  \[
    \mu_X \equiv \big(
    \begin{tikzcd}
        M×(M×X)
        \arrow{r}{\alpha_{M,M,X}^{-1}}
        & (M×M)×X
        \arrow{r}{m×X}
        & M×X
    \end{tikzcd}
    \big).
  \]
  \begin{lemma}\label{lem:can-distrib}
    Given a monoid $(M,m,e)$ and a \Set-functor $F$ with strength $s$, the
    natural transformation
    \begin{equation*}
      s_{M,X}:M\times FX \to F(M\times X)
    \end{equation*}
    is a distributive law of the monad $M\times(-)$ over the
    functor $F$.
  \end{lemma}
  \begin{proof}
    Using $(F\iota_X)^{-1} = F\iota_X^{-1}$, the commutatitivity of the
    following verifies \eqref{diag:distrib-unit}:
    \[
    \begin{tikzcd}
        &FX
        \arrow{dl}[above left]{\iota_{FX}^{-1}}
        \arrow{dr}{F\iota_{X}^{-1}}
        \descto[pos=0.4]{d}{\eqref{diag:strength1}}
        \arrow[to path={
                -| ([xshift=-4mm]\tikztotarget.west) \tikztonodes
                -- (\tikztotarget)
              }]{ddl}[pos=0.8,left]{\eta_{FX}}
        \arrow[to path={
                -| ([xshift=4mm]\tikztotarget.east) \tikztonodes
                -- (\tikztotarget)
              }]{ddr}[pos=0.8,right]{F\eta_{X}}
        \\
        1×FX
        \arrow{d}[left]{m×FX}
        \arrow{rr}{s_{1,X}}
        \descto{drr}{\text{Naturality}}
        & {}
        & F(1×X)
        \arrow{d}{F(m×X)}
        \\
        M×FX
        \arrow[below]{rr}{s_{M,X}}
        &
        & F(M×X)
    \end{tikzcd}
    \]
    Also using $F(\alpha^{-1}_{M,M,X})=(F\alpha_{M,M,X})^{-1}$, note
    that
    \[
    \begin{tikzcd}[column sep = 8mm]
        M×(M×FX)
        \arrow{r}[pos=.9]{\begin{turn}{45}$\scriptsize M×s_{M,FX}$\end{turn}}
        \arrow{d}[left]{\alpha_{M,M,FX}^{-1}}
        \descto{drr}{\eqref{diag:strength2}}
        \arrow[shiftarr={xshift=-16mm}]{dd}[left]{\mu_{FX}}
        &
        M×F(M×X)
        \arrow{r}[pos=.9]{\begin{turn}{45}$\scriptsize s_{M,F(M×X)}$\end{turn}}
        & F(M×(M×X))
        \arrow{d}{F\alpha_{M,M,X}^{-1}}
        \arrow[shiftarr={xshift=16mm}]{dd}{F\mu_{X}}
        \\
        (M×M)×FX
        \arrow{d}[left]{m×FX}
        \arrow{rr}{s_{M×M,X}}
        \descto{drr}{\text{Naturality}}
        &&
        F((M×M)×X)
        \arrow{d}{F(m×X)}
        \\
        M×FX
        \arrow{rr}[below]{s_{M,X}}
        &&
        F(M×X)
    \end{tikzcd}
    \]
    commutes, so $s_{M,\argument}$ satisfies \eqref{diag:distrib-mult}, and
    hence is a distributive law.
  \end{proof}
  \begin{definition}
    For $M=\perms(\V)$, we refer to the distributive law described in
    Lemma~\ref{lem:can-distrib} as the \emph{canonical distributive law}
    of $\perms(\V)\times(-)$\, over $F$, and to the arising
    $\perms(\V)$-set lifting of $F$ as the \emph{canonical
      $\perms(\V)$-set lifting of $F$}.
  \end{definition}
  \begin{lemma}\label{lem:can-lifting-nom}
    The canonical $\perms(\V)$-set lifting of a finitary \Set-functor
    is $\Nom$-restricting.
  \end{lemma}
  \begin{proof}
    Let $F:\Set\to\Set$ be finitary, let $\bar F$ denote the canonical
    $\perms(\V)$-set lifting of $F$, let $X$ be a nominal set with
    nominal structure $\alpha: \perms(\V)×X \to X$, and let $x\in
    FX$.
    We have to show that $x$ has finite support in $\bar FX$. Since
    $F$ is finitary, $x: 1\to FX$ factors through some $Fi$ with $i$ a
    subset inclusion $S\hookrightarrow X$ of a finite subset $S$.

    Then $W=\supp(S)$ supports $x$: Put $G=\perms(\V\setminus W)$, let
    $m: G\to \perms(\V)$ be the evident subgroup inclusion, and let
    $\pi \in G$. Since $W$ supports $S$ and $S$ is finite, the
    elements of $G$ fix $S$ pointwise, i.e.\
    \begin{equation}\label{diag:G-fix-S}
    \begin{tikzcd}
      G× S \arrow{r}{m×i} \arrow[shiftarr={yshift=5mm}]{rr}{i\circ
        \outr} & \perms(\V)×X \arrow{r}{\alpha} & X
    \end{tikzcd}
    \end{equation}
    commutes, where $\outr$ denotes the right-hand product projection.  With
    $\beta$ denoting the nominal structure on $FX$ and $s$ the
    strength (so $s_{\perms(\V),\argument}$ is the canonical
    distributive law), we have that
    \begin{equation*}
    \begin{tikzcd}
      1\cong 1×1 \arrow{r}{1× x} \arrow[bend
      right=40]{dddrr}[sloped,below]{\pi \cdot x} \arrow[to path={--
        ([yshift=2mm]\tikztostart.north) -|
        ([xshift=14mm,yshift=-1mm]\tikztotarget.east) \tikztonodes --
        ([yshift=-1mm]\tikztotarget.east) } ]{rrddd}[above
      right,pos=0.8]{x} & 1× FS \arrow{r}{s_{1,S}} \arrow{d}[left]{\pi
        × FS} \descto{dr}{\text{Naturality of }s} & F(1×S)
      \arrow{d}{F(\pi×S)} \\ & G×FS \arrow{r}{s_{G,S}}
      \arrow{d}[left]{m×Fi} \descto{dr}{\text{Naturality of }s} &
      F(G×S) \arrow{d}{F(m×i)}
      \arrow[shiftarr={xshift=12mm}]{dd}[sloped,above]{F(i\circ
        \outr)} \\ & \perms(\V)×FX \arrow{r}{s_{\perms(V),X}}
      \arrow{dr}[below left]{\beta}
      \descto[xshift=2mm,yshift=2mm]{dr}{\text{Def.}}  &
      F(\perms(\V)×X) \arrow{d}{F\alpha} \\ & & FX
    \end{tikzcd}
  \end{equation*}
  commutes, where the unlabelled triangle commutes by
  Diagram~(\ref{diag:G-fix-S}) and the decomposition of $x$ in the
  upper right hand part is by the strength
  law~(\ref{diag:strength1}). This shows that $\pi\cdot x = x$, as
  required.
\end{proof}

\begin{definition}
  We refer to the lifting of a \Set-functor $F$ to $\Nom$ arising from
  Lemma~\ref{lem:can-lifting-nom} as the \emph{canonical lifting of
    $F$}. Moreover, a lifting $G:\Nom\to\Nom$ is \emph{canonical} if
  it is a canonical lifting of some functor (i.e.~a canonical lifting
  of $UGD$).
\end{definition}
\noindent The canonical lifting is the expected lifting for many
\Set-functors:
\begin{example}\label{setFunctorList}
  %We explicitly describe the canonical lifting for a number of
  %finitary set functors of interest.
  \begin{enumerate}[(1)]
  \item For a polynomial functor $F_\Sigma$ on $\Set$ (see Example~\ref{ex:coalgs}(3)) the canonical lifting $\bar F_\Sigma$ maps a nominal set
    $(X, \cdot)$ to the expected coproduct of finite products in
    $\Nom$ where each $\Sigma_n$ is equipped with the trivial nominal
    structure.
  \item The canonical lifting of the finite powerset functor $\Potf$
    maps a nominal set $(X, \cdot)$ to $\Potf(X)$ equipped with the
    usual nominal structure, which is given by
    $\pi \cdot Y = \{\pi \cdot y \mid y \in Y\}$ for $Y\in\Potf(X)$.
  \item The canonical lifting of the bag functor $\bar \Bag$ maps a nominal
    set $(X, \cdot)$ to $\Bag(X)$ equipped with the nominal structure
    that acts elementwise as in the previous item.
  \item An interesting more general class of examples are Joyal's \emph{analytic functors}~\cite{Joyal:1981,Joyal:1986}. An endofunctor $F$ on \Set is  \emph{analytic} if it is the left Kan extension of a functor from the category $\mathsf{B}$ of natural numbers and bijections to $\Set$ along the inclusion. These are described explicitly as follows. For a subgroup $G$ of $\sym(n)$, $n < \omega$, the \emph{symmetrized representable functor} maps a set $X$ to the set $X^n/G$ of orbits under the action of $G$ on $X^n$ by coordinate interchange, i.e., $X^n/G$ is the quotient of $X^n$ modulo the equivalence $\sim_G$ with $(x_0,\dots,x_{n-1})\sim_G (y_0,\dots,y_{n-1})$ iff $(x_{\pi(0)},\dots,x_{\pi(n-1)})=(y_0,\dots,y_{n-1})$ for some $\pi\in G$. It is not difficult to prove that an endofunctor on $\Set$ is analytic iff it is a coproduct of symmetrized representables. So every analytic functor $H$ can be written in the form
    \begin{equation}
      \label{functor_from_symmetrized_representables}
      FX=\coprod_{\substack{n< \omega\\ \mathclap{G\leq \perms(n)}}} A_{n,G}\times X^n/G \text{ .}
    \end{equation}
    Clearly every analytic functor is finitary, and Joyal proved in \cite{Joyal:1981,Joyal:1986} that a finitary endofunctor on $\Set$ is analytic iff it weakly preserves wide pullbacks.

    The canonical lifting of an analytic functor $F$ is given by equipping for any nominal set $(X,\cdot)$ the quotients $X^n/G$ with the obvious group action: 
    \[
      \pi \cdot [(x_0, \ldots, x_{n-1}]_{\sim_G} = [\pi \cdot x_0, \ldots, \pi  \cdot x_{n-1}]_{\sim_G}.
    \]
    
    Note that the bag functor from the previous item is the special case where we take $A_{n,G} = 1$ for $G = \perms(n)$ and $0$ else, for every $n$. The finite power-set functor is not analytic. 

  \item Another interesting analytic functor is the \emph{cyclic shift
      functor} $\mathcal Z$ that maps a set $X$ to the set of all
    assignments of elements of $X$ to the corners of any regular
    polygon (modulo rotation of the polygon). In fact, this is the
    analytic functor obtained by putting $A_{n,G} = 1$ for $G$
    generated by the cyclic right shift $\pi(i) = (i + 1) \mod n$ for
    $i = 0, \ldots, n-1$, and $A_{n,G} = 0$ otherwise. The canonical
    lifing of $\mathcal Z$ is as expected: given a nominal set $Y$,
    the nominal structure on $\bar{\mathcal Z} Y$ acts by applying the
    original action on $Y$ to the elements labelling the corners of a
    regular polygon.
  \end{enumerate}
\end{example}

\noindent However, many important functors on $\Nom$ are liftings but
not canonical liftings.
\begin{example}\label{expl:non-canonical}
  \begin{enumerate}[(1)]
  \item The simplest examples are constant functors
    $\bar KX=(Y,\cdot\,)$ for a nontrivial nominal set $(Y,\cdot\,)$. The
functor $\bar K$ is clearly a lifting of the constant functor $KX=Y$
on $\Set$ but the canonical lifting of $K$ is constantly $DY$. 

  \item Similarly, more interesting composite functors such as the functor $\bar LX = \V + X \times X + \V \times X$ mentioned in the introduction with the standard action on $\V$, i.e.~$\pi \cdot v = \pi(v)$ for $v \in \V$, or the functors $\bar{\mathcal B}(-) + \V$ or $\bar{\mathcal Z}(-) + \V$ are non-canonical liftings. 
  \end{enumerate}
\end{example}
\noindent We therefore identify a property of distributive
laws that (in combination with restriction of liftings to $\Nom$)
suffices to enable our main result on rational fixpoints of liftings:
  \begin{definition}\label{def:localizable}
    Let $F: \Set\to \Set$ be a functor. A monad-over-functor
    distributive law $\lambda: \perms(\V) × F(-) \to F(\perms(\V)× -)$
    is \emph{localizable} if for any $X$ and any $W \subseteq \V$,
    $\lambda$ restricts to a natural transformation
    $\lambda^{W}: \perms(W) × FX \to F(\perms(W)×(-)\,)$,
    i.e.~we have
    $\lambda_X \cdot (m_W × \id_{FX}) = F(m_W×\id_X) \cdot
    \lambda_X^{W}$
    where $m_W: \perms(W)\rightarrow \perms(\V)$ is the evident
    subgroup inclusion. A lifting $G:\Nom\to\Nom$ is
    \emph{localizable} if it is induced by a $\Nom$-restricting
    localizable distributive law.
  \end{definition}
  \begin{lemma}
    Canonical liftings are localizable.
  \end{lemma}
  \begin{proof}
    The equation
    $\lambda_X \cdot (m_W × \id_{FX}) = F(m_W×\id_X) \cdot
    \lambda_X^{W}$
    postulated in Definition~\ref{def:localizable} is an instance of
    naturality of the strength in the left argument.
  \end{proof}
  \begin{example}\label{expl:const-localizable}
    By the previous lemma, in particular the identity functor on \Nom
    is a localizable lifting. Moreover, all constant functors on \Nom
    are trivially localizable.

  \end{example}
  \begin{lemma}\label{lem:localizable-closure}
    The class of finitary and mono-preserving localizable liftings is
    closed under finite products, arbitrary coproducts, and functor
    composition.
  \end{lemma}
  \begin{proof}
    Finitarity and preservation of monos are clear. The lifting
    property is immediate from creation of finite products and
    coproducts by $U: \Nom \to\Set$.  Since moreover finite products
    and coproducts in $\Nom$ are formed as in $\GSet$, it is clear
    that they preserve the property of being induced by a
    $\Nom$-restricting distributive law. It remains to show
    preservation of localizability by the mentioned constructions. Let
    $W\subseteq\V$.

    \emph{Finite products:} Since the terminal functor is constant, it
    suffices to consider binary products $G\times H$ of functors $G,H$
    on $\Nom$ induced by \Nom-restricting localizable distributive
    laws $\lambda_G$, $\lambda_H$. Indeed, $G\times H$ is induced by
    the distributive law
    \begin{equation*}
      (\lambda_{G\times
        H})(\pi,(x,y))=((\lambda_G)_X(\pi,x),(\lambda_H)_X(\pi,y)),
    \end{equation*}
    and if $\pi\in\perms(W)$ then the right-hand side is in
    $(G\times H)(\perms(W)\times X)$.

    \emph{Coproducts:} For $i\in I$, let the functors $G_i$ on $\Nom$
    be induced by \Nom-restricting localizable distributive laws
    $\lambda_{G_i}$. Then $G=\coprod_{i\in I}G_i$ is induced by
    the distributive laws
    \begin{equation*}
      (\lambda_{G})_X(\pi,\inj_i(x))=\inj_i((\lambda_{G_i})_X(\pi,x)),  
    \end{equation*}
    where $\inj_i$ denotes the $i$-th coproduct injection, and if
    $\pi\in\perms(W)$ then the right-hand side is in
    $\coprod G_i(\perms(W)\times X)$.

    \emph{Functor composition:} Let $G,H$ be functors on $\Nom$
    induced by \Nom-restricting localizable distributive laws
    $\lambda_G$, $\lambda_H$. Then $GH$ is induced by the distributive
    law 
    \begin{equation*}
      (\lambda_{GH})_X(\pi,x)=G(\lambda_H)_{X}((\lambda_G)_{HX}(\pi,x)),
    \end{equation*}
    and if $\pi\in\perms(W)$ then the right-hand side is in
    $GH(\perms(W)\times X)$.
  \end{proof}
\begin{definition} \label{def:polynomial}
Recall that the class of \emph{polynomial functors} is the
smallest class of endofunctors on $\Nom$ that contains all constant functors and the identity
functor and is closed under coproducts and finite products. 
\end{definition}
\noindent Note in particular that the functors in
Example~\ref{expl:non-canonical} are polynomial. By
Example~\ref{expl:const-localizable} and
Lemma~\ref{lem:localizable-closure}, we have
\begin{lemma}
  All polynomial functors are finitary localizable liftings and preserve monos.
\end{lemma}
\noindent We see next that there are liftings that fail to be
localizable, and indeed our example does not allow the desired lifting
of rational fixpoints from \Set to~\Nom:
\begin{example}\label{nonLocalExample}
  Consider the functor $FX = \V × X$ with the lifting
  $\tilde F (X,\cdot) = (\V,\cdot) × (X,\star)$, where $\cdot$ is the
  usual action on $\V$ and $\pi\star x$ is defined as
  $g\cdot \pi \cdot g^{-1}\cdot x$ for some fixed permutation
  $g: \V\to \V$ such that there is a name $v_0\in\V$ for which the
  names $g^n\cdot v_0=:v_n$ are pairwise distinct for $n\in\Z$. This
  is well-defined because by Remark~\ref{yetAnotherRemark}, $(X,\cdot)$ uniquely
  extends to a $\sym(\V)$-set. This lifting corresponds to the
  distributive law defined by
    \[
        \lambda_X : \perms(\V) × \V×X \to \V×\perms(\V) × X,
        \quad
        (\pi,v,x) \mapsto (\pi(v), g\cdot\pi\cdot g^{-1}, x).
    \]
    This distributive law does not satisfy locality; to see this,
    consider $W = \{v_0, v_1\}$ and $\pi = (v_0\, v_1)\in\perms(W)$;
    then $g \cdot \pi \cdot g^{-1}=(v_1\,v_2)$ is not in $\perms(W)$
    (qua subgroup of $\perms(\V)$), since $v_2\notin W$.

    We have rational fixpoints $\varrho\tilde F$ in $\Nom$ and
    $\varrho F$ in $\Set$; we show that (1) neither is
    $U\varrho\tilde F$ a subcoalgebra of $\varrho F$, and (2) nor does
    $\varrho F$ lift to an $\tilde F$-coalgebra:
    \begin{enumerate}[(1)]
    \item The rational fixpoint $\varrho \tilde F$ contains behaviours
      that are not lfp in \Set, so that $U(\varrho\tilde F)$ is not a
      subcoalgebra of $\varrho F$: consider the coalgebra
      $c: (\V,\cdot)\to \tilde F(\V,\cdot)= (\V,\cdot)×(\V,\star)$
      defined by $c(v) = (v,g(v))$. This coalgebra structure is
      equivariant, because
    \[
        c(\pi\cdot v)
        = (\pi \cdot v, g\cdot\pi\cdot v))
        = (\pi \cdot v, g\cdot\pi\cdot g^{-1}\cdot g\cdot v))
        = (\pi\cdot v, \pi\star g(v)).
    \]
    Since $\V$ is orbit-finite, $c$ is~lfp. Moreover, $c$ is a
    subcoalgebra of $\varrho \tilde F$, i.e.~the coalgebra
    homomorphism $c^\dagger: (\V,c)\to \varrho (\tilde F, r)$ is
    monic, because $v$ can be recovered from $c^\dagger(v)$ via
    $v = \outl\circ \tilde Fc^\dagger \circ c(v) = \outl\circ r \circ
    c^\dagger (v)$.

    However, $(\V,c)$ considered as an $F$-coalgebra in \Set is not
    lfp, because the smallest subcoalgebra containing $v_0$ is
    $\{g^n\cdot v_0\mid n\ge 0\}$, which is infinite by the choice of
    $g$ and $v_0$.

  \item The coalgebra $d: 1\to F1$ defined by
    $d(*) = (v_0,*) \in \V × 1$ (where $1=\{*\}$) is, trivially, lfp
    and a subcoalgebra of $\varrho F$.  The unique coalgebra
    homomorphism $1 \to \varrho F$ defines an element
    $d^\dagger \in \varrho F$.

    Assuming some $\perms(\V)$-set structure $\cdot$ on $\varrho F$
    such that $r: (\varrho F,\cdot) \to \tilde F (\varrho F,\cdot)$ is
    equivariant, the support of $d^\dagger$ must contain $v_0$ and the
    support of $d^\dagger$ in $(\varrho F,\star)$, which is the
    support of $g^{-1}\cdot d^\dagger$ in $(\varrho F,\cdot)$.
    Iterating this observation we find that the support of $d^\dagger$
    contains $g^{-n}\cdot v_0$ for all $n\in\N$, hence is infinite.
    \end{enumerate}
\end{example}
\begin{remark}
  The functor $\tilde F$ from the previous example is naturally isomorphic to
  the harmless (in fact, polynomial) functor
    \[
        \bar F (X,\cdot) = (\V,\cdot) × (X,\cdot),
    \]
    where the isomorphism $\tau: \bar F \to \tilde F$ is given by
    \[
        \tau_X (v,x) = (v, g\cdot x)
    \]
    using the fact that the action of $\perms(\V)$ on the nominal set
    $X$ extends uniquely to an action of $\sym(\V)$
    (Remark~\ref{yetAnotherRemark}.\ref{item:inf-permute}). In fact,
    $\tau_X$ is clearly bijective. We have to show that $\tau_X$ is
    equivariant:
    \[
        \tau_X(\pi\cdot v, \pi\cdot x)
            = \big(\pi\cdot v, g\cdot \pi \cdot x\big)
            = \big(\pi\cdot v, g\cdot \pi \cdot g^{-1} \cdot g\cdot x\big)
            = \big(\pi\cdot v, \pi\star (g\cdot x)\big).
    \]
    Finally, we compute the naturality square for an equivariant map
    $f: X\to Y$:
    \[
        \tau_Y\cdot (\id_\V × f)(v,x)
        = (v,g\cdot f(x))
        = (v,f(g\cdot x))
         = (\id_\V×f)\cdot \tau_X(v,x).
    \]

    This fact may be slightly surprising, and shows that lifting of
    the rational fixpoint is a representation-dependent property of
    functors on $\Nom$ rather than an intrinsic one; it serves only as
    a technical tool in the computation of rational fixpoints. The
    isomorphism $\tilde F\cong \bar F$ implies that the counterexample
    to localizability is not only somewhat contrived but can also be
    circumvented; that is, instead of calculating the rational
    fixpoint of $\tilde F$ we can calculate that of $\bar F$, which is
    perfectly amenable to our methods. In fact we have no example of a
    lifting that is not isomorphic to a localizable one.
\end{remark}
\noindent In Example~\ref{nonLocalExample}, we have seen that
\Nom-restricting distributive laws need not be localizable. As the
following simple example shows, localizability and $\Nom$-restriction
are in fact independent, i.e.~localizable distributive laws also need
not be $\Nom$-restricting:
\begin{example}\label{exampleNonNominal}
  Let $(Y,\cdot)$ be some non-nominal $\perms(\V)$-set. Consider the
  constant functor $KX = Y$ with the distributive law
  $\lambda_X: \perms(\V)×KX \to K(\perms(\V)×X)$ defined by
    \[
    \lambda_X(\pi,y) = \pi\cdot y\qquad (y\in Y).
    \]
    Since $K$ is constant, $\lambda$ is trivially
    localizable. However, $\lambda$ induces, as its $\perms(\V)$-set
    lifting, the constant functor $\bar K=(Y,\cdot)$, and hence fails
    to be $\Nom$-restricting.
\end{example}

\section{Rational Fixpoints of Localizable Liftings}
\label{sec:lifting-rational}

We proceed to analyse rational fixpoints of liftings $\bar F$ in
relation to rational fixpoints of the underlying functor $F$.  We have
seen in Example~\ref{nonLocalExample} that even when $F$ is finitary,
the rational fixpoint of $\bar F$ in general need not be a lifting of
the rational fixpoint of $F$ (in contrast to the situation with
initial algebras). Our main result (Theorem~\ref{thm:lift}) establishes
that for \emph{localizable} liftings, the rational fixpoint of $F$
does lift to the rational fixpoint of $\bar F$. As a consequence, we
also obtain concrete descriptions of the rational fixpoint for
functors on $\Nom$ that are quotients of lifted functors (but not
themselves liftings of \Set-functors) (Section~\ref{sec:quotients}),
e.g.~functors associated to a binding signature
(Section~\ref{sec:binding}).

\begin{assumption} \label{lambdaassumption} In this section, assume
  that $\bar F: \Nom\to\Nom$ is a localizable lifting of a finitary
  functor $F: \Set\to\Set$.
\end{assumption}

\begin{lemma}\label{setLfp2NomLfp}
    If for a coalgebra $c: C\to \bar FC$ the underlying coalgebra
    $c: C\to FC$ is lfp in \Set, then $c: C\to \bar FC$  is lfp in \Nom.
\end{lemma}
\begin{proof}
  Let $x\in C$, and let $O$ be the orbit of $x$; we have to construct
  an orbit-finite subcoalgebra $Q$ of $C$ containing $x$. The lfp
  property of $(C,c)$ in \Set provides us with a finite subcoalgebra
  $(P, p)$ with $x \in P$.  We take $Q\subseteq C$ to be the union of
  the orbits of the elements of $P$, i.e.~the closure of $P$ in $C$
  under the $\perms(\V)$-action. Then $Q$ is a nominal set; applying
  $\bar F$ to the equivariant inclusion $Q\to C$, we obtain that $FQ$
  is closed under the $\perms(\V)$-action in $FC$. Note also that $Q$
  is orbit-finite since $P$ is finite. % Now
  We are done once show that $Q$ is closed under the coalgebra
  structure~$c$. Let $y\in P$ and $\pi\in\perms(\V)$, so that
  $\pi\cdot y\in Q$. Since $P$ is a subcoalgebra of $(C,c)$ in $\Set$,
  we have $p(y) = c(y)$ and hence
    \[
        c(\pi \cdot y) = \pi \cdot c(y) = \pi\cdot p(y).
    \]
    Since $p(y) \in FP\subseteq FQ$ and $FQ$ is closed under the
    $\perms(\V)$-action in $FC$, it follows that
    $c(\pi \cdot y) \in FQ$.
\end{proof}

\begin{lemma}\label{nomLfp2SetLfp}
    If $c: C\to \bar FC$ is an orbit-finite coalgebra in $\Nom$, then $c: C\to FC$ is
    lfp in \Set.
\end{lemma}
\begin{proof}
    First define the following closure operator on subsets $X$ of $C$:
    \[
        \Cl(X) = \{
            y\in C \mid \supp(y) \subseteq\textstyle \bigcup_{x\in X}\supp(x)
        \}.
    \]
    By Lemma~\ref{sameorbitfaculty}, $\Cl$ preserves finite sets. Now
    let $x\in C$. Pick a subset $O\subseteq C$ that contains precisely
    one element from each orbit of $C$, and put
    $P =\Cl(\{x\} \cup O)$, with $\inj_P$ denoting the embedding
    $ P \rightarrowtail C$. Since $O$ is finite, $P$ is finite, and
    since $F$ is finitary, there exists a finite set $Q\subseteq C$
    such that $c\cdot \inj_P$ factorizes through
    $F\inj_Q : FQ \to FC$. The subset $W = \Cl(P\cup Q) \subseteq C$
    is finite as well, and we have:
    \[
    \begin{tikzcd}
        C
            \arrow{r}{c}
        & FC
        \\
        P
            \arrow[>->]{u}{\inj_P}
            \arrow[dashed]{r}{c_P}
            \arrow[dashed]{dr}
        &
        FW
        \arrow[>->]{u}{F\inj_W}
        \\ &
            FQ
        \arrow[>->]{u}
        \arrow[>->,shiftarr={xshift=8mm}]{uu}[right]{F\inj_Q}
    \end{tikzcd}
    \]
    Now let $G$ be the finite subgroup of $\perms(\V)$ given by the
    permutations of $\supp(W)$ and note that
    $\supp(W) = \supp(P \cup Q)$. Then for any $\pi \in G$ and
    $z\in P$, $\pi\cdot z$ is in $W$ because
    \[
      \supp(\pi \cdot z) = \pi \cdot \supp(z) \subseteq \pi \cdot \supp(W) = \supp(W) = \supp(P \cup Q),
    \]
    where the second-to-last equation holds because $\pi \in G$ and the inclusion holds because $z \in P \subseteq W$. This means we have a commutative diagram
    \[
      \begin{tikzcd}
        G \times P \arrow[->>]{d}[left]{\alpha'} \arrow[>->]{r} & \perms(\V) \times C \arrow{d}{\alpha} \\
        W \arrow[>->]{r}[below]{i} & C
      \end{tikzcd}
    \]
    where $\alpha$ denotes the nominal structure on $C$. We will now
    prove that the left-hand map $\alpha'$ is surjective. To see this,
    let $y\in W$. Then there are $z \in O\subseteq P$ and
    $\pi \in \perms(\V)$ such that $\pi\cdot z = y$, because $O$
    contains precisely one element from each orbit. Consider the
    factorization of $\pi$ into $\pi = \pi\vert_{\supp(z)}\cdot g$ as
    in Remark~\ref{permFactor}.  Then $g$ fixes every element of
    $\supp(z)$, so $g\cdot z = z$, and $\pi\vert_{\supp(z)}$ fixes
    every element not contained in $\supp(z) \cup \pi \cdot
    \supp(z)$.
    Since $\pi \cdot z = y$ we have $\pi \cdot \supp(z) = \supp(y)$,
    and because $y \in W$ and $z \in P \subseteq W$ we know that
    $\supp(z) \cup \supp(y) \subseteq \supp(W)$. Thus
    $\pi\vert_{\supp(z)}$ fixes every element not contained in
    $\supp(W)$ and therefore $\pi\vert_{\supp(z)}$ lies in $G$. It
    follows that
    \[
        \alpha'(\pi\vert_{\supp(z)}, z)
        = \pi\vert_{\supp(z)} \cdot z
        = \pi\vert_{\supp(z)} \cdot g \cdot z
        = \pi \cdot z
        = y,
    \]
    showing $\alpha'$ to be surjective as desired. 

    Now fix a splitting \( d: W \rightarrowtail G × P\)
    of $\alpha'$, i.e.~we have $\alpha' \cdot d = \id_W$.  Denote by
    $\alpha'': G \times W \to C$ the restriction of $\alpha$. Let
    $\beta = F\alpha \cdot \lambda_C$ be the nominal structure on $FC$
    and $\beta': G×FW \to FC$ its restriction. Now consider the
    diagram below:
    \[
    \begin{inparaenum}[(1)]
    \begin{tikzcd}
        C
        \arrow{r}{c}
        \descto{ddr}{\text{\item}}
        & FC
        \arrow[equals]{r}
        \descto{ddr}{\text{\item}}
        & FC
        \arrow[equals]{rr}
        &{}
        & FC
        \\
        &&& F(G × W)
        \arrow{ul}[above right]{F\alpha''}
        \\
        G × P
        \arrow{r}[below]{G× c_P}
        \arrow{uu}{i \cdot \alpha'}
        &
        G × FW
        \arrow{r}[below]{\lambda_W^{G}}
        \arrow{uu}{\beta'}
        &
        F(G × W)
        \arrow[>->]{r}[below,pos=0.9]{\begin{turn}{-45}$\scriptsize F(G× d)$\end{turn}}
        \arrow{uu}{F \alpha''}
        \arrow[>->]{ur}[above left]{\id}
        &
        F(G^2 × P)
        \arrow{r}[below,pos=0.9]{\begin{turn}{-45}$\scriptsize F(\mu×P)$\end{turn}}
        \arrow{u}[right]{F(G×\alpha')}[left,xshift=-5pt,pos=0.4]{\text{\item}}
        &
        F(G × P)
        \arrow{uu}[right]{F(i \cdot \alpha')}
        \descto[xshift=5mm]{uul}{\text{\item}}
    \end{tikzcd}
    \end{inparaenum}
    \]
    The middle triangle trivially commutes, and so do the other parts:
    \begin{enumerate}[(1)]
    \item commutes because $c$ is equivariant.
    \item commutes using the definition of $\beta$, naturality of $\lambda$ and~Assumption~\ref{lambdaassumption} (denote by $j: G \to \perms(\V)$ the inclusion of the subgroup $G$):
      \[
        \begin{tikzcd}
          FC \arrow[equals]{r} & FC \\
          F(\perms(\V) \times C) 
          \arrow{u}[left]{F\alpha}
          &
          \\
          \perms(\V) \times FC
          \ar{u}[left]{\lambda_C}
          \ar{r}[below]{\lambda_C}
          &
          F(\perms(\V) \times C) 
          \ar{uu}[right]{F\alpha}
          \\
          \perms(\V) \times FW
          \ar[>->]{u}{\id \times Fi}
          \ar{r}[below]{\lambda_W}
          &
          F(\perms(\V) \times W)
          \ar{u}[right]{F(\id \times i)}
          \\
          G \times FW 
          \ar{u}[left]{j \times \id}
          \ar{r}[below]{\lambda_W^G}
          \ar[shiftarr={xshift=-16mm}]{uuuu}[left]{\beta'}
          &
          F(G \times W)
          \ar{u}[right]{F(j \times \id)}
          \ar[shiftarr={xshift=16mm}]{uuuu}[right]{F\alpha''}
        \end{tikzcd}
      \]
      
    \item commutes since $\alpha' \cdot d = \id_W$.
    \item commutes using the axioms of the group action $\alpha$; here $\mu$ denotes the multiplication of the group $G$ and we also use that $G$ is a subgroup of $\perms(\V)$.
    \end{enumerate}
    Thus, we see that $G \times P$ is a finite coalgebra and $i \cdot \alpha': G \times P \to C$ a coalgebra homomorphism with $\alpha'(\id, x) = x$. Therefore $x \in C$ is contained in a finite subcoalgebra and we conclude that $(C,c)$ is lfp.   
\end{proof}

\noindent Because $U: \Nom\to \Set$ creates, and lfp coalgebras are
closed under, filtered colimits, we can immediately generalize
Lemma~\ref{nomLfp2SetLfp} to lfp coalgebras:
\begin{corollary}
    If $c: C\to \bar FC$ is lfp in $\Nom$, then $c: C\to FC$ is lfp in \Set.
\end{corollary}
\noindent Combining this with Lemma~\ref{setLfp2NomLfp} we obtain:
\begin{corollary}
    A coalgebra $c: C\to \bar FC$ in $\Nom$ is lfp if and only if the underlying
    coalgebra is lfp in \Set.
\end{corollary}
\noindent In order to lift the rational fixpoint $(\varrho F, r)$ from
\Set to \Nom, we need to equip it with nominal structure:
\begin{lemma}\label{gsetstructure}
    The rational fixpoint $(\varrho F, r)$ carries a canonical group action making $r$ equivariant.
\end{lemma}
\begin{proof}
  We define the desired group action by coinduction. To this end we consider the $F$-coalgebra 
  \begin{equation*}
    \perms(\V)×\varrho F \xrightarrow{\id × r} \perms(\V) \times F(\varrho F) \xrightarrow{\lambda_{\varrho F}}
    F(\perms(\V) \times \varrho F).
  \end{equation*}
  We first prove that this coalgebra is lfp.  Let
  $(\pi,x) \in \perms(\V)\times \varrho F$. Since $\varrho F$ is an
  lfp coalgebra, we obtain an orbit-finite subcoalgebra $(S,s)$ of
  $\varrho F$ containing $x$. For the finite subgroup
  $G = \perms(\supp(\pi))$ of $\perms(\V)$, we have a restriction
  $\lambda_S^G:G×FS \to F(G× S)$ by localizability. Now consider the
  diagram below (where, as usual, we abuse objects to denote their
  identity, here: $G$ in place of $\id_G$):
  \begin{equation*}
    \begin{inparaenum}[(1)]
\begin{tikzcd}
        & \perms(\V) × \varrho F \arrow{r}{\ \perms(\V)× r^F}
        \descto[xshift=4mm]{dr}{\text{\item}} &[6mm] \perms(\V) × F \varrho
        F \arrow{r}{\lambda_{\varrho F}}
        \descto[xshift=4mm]{dr}{\text{\item}} & F(\perms(\V) × \varrho
        F)
        \\
        1 \arrow[bend left]{ur}[above left]{(\pi,x)}
        \arrow{r}[above]{(\pi,x)} \arrow[bend right]{dr}[below
        left]{(\pi,x)} & \perms(\V) × S
        \arrow{u}[right]{\perms(\V)×\inj_S} \rar{\perms(\V)×s}
        \descto[xshift=3.1mm,yshift=0.3mm]{dr}{\text{\item}} &[6mm]\perms(\V) × FS
        \arrow{u}[right]{\perms(\V)×F\inj_S} \rar{\lambda_S}
        \descto[xshift=4mm]{dr}{\text{\item}} & F(\perms(\V)× S)
        \arrow{u}[right]{F(\perms(\V)×\inj_S)}
        \\
        & G× S \arrow{u}[right]{\inj_{G}×S} \rar{G×s} &[6mm] G× FS
        \arrow{u}[right]{\inj_{G}×FS} \rar{\lambda_S^{G}} & F(G×S)
        \arrow{u}[right]{F(\inj_{G}×S)}
      \end{tikzcd}
    \end{inparaenum}
  \end{equation*}
  It commutes because all its inner parts do:
  \begin{enumerate}[(1)]
  \item $(S,s)$ is a subcoalgebra of $(\varrho F, r^F)$.
  \item Naturality of $\lambda$ .
  \item Properties of products.
  \item $\lambda_S^{G}$ restricts $\lambda_S$.
  \end{enumerate}
  Hence $(G\times S, \lambda^G_S \cdot (G \times s))$ is a finite
  subcoalgebra of $\perms(\V) \times \varrho F$ containing
  $(\pi, x)$ proving $\perms(\V) \times \varrho F$ to be an lfp
  coalgebra.
  
  Now we obtain a unique coalgebra homomorphism
  $u: \perms(\V)×\varrho F \to \varrho F$. It remains to show that $u$
  is a group action. To this end, we show that $u$ is the restriction
  of the group action on the final $F$-coalgebra to $\varrho F$.

  From \cite[Theorem 3.2.3]{bartels04} and \cite{plotkinturi97} we
  know that the final $F$-coalgebra in the $\perms(\V)$-sets is just
  the final $F$-coalgebra $(\nu F, t)$, with the group action on the
  carrier defined by coinduction, i.e., the group action is the unique
  map $a$ such that the diagram below commutes:
  \begin{equation*}
    \begin{tikzcd}
        \perms(\V)×\nu F
            \rar{\perms(\V)× t}
            \arrow[dashed]{d}[left]{a}
        &[4mm] \perms(\V)×F\nu F
            \rar{\lambda_{\nu F}}
        & F(\perms(\V)× \nu F),
            \arrow{d}[right]{Fa}
        \\
        \nu F \arrow{rr}{t}
        && F\nu F.
    \end{tikzcd}
  \end{equation*}
  Since $F$ preserves monos, the rational fixpoint is a subcoalgebra
  of $(\nu F, t)$, i.e.~the unique coalgebra homomorphism
  $j: (\varrho F,r)\rightarrowtail (\nu F, t)$ is monic. Then
  $\perms(\V)×j$ also is a coalgebra homomorphism:
  \begin{equation*}
    \begin{inparaenum}[(1)]
\begin{tikzcd}
        \perms(\V) × \varrho F \arrow{rr}{\perms(\V)× r}
        \arrow{d}[left]{\perms(\V)×j}
        \descto[xshift=-2mm]{drr}{\text{\item}} && \perms(\V) × F
        \varrho F \arrow{r}{\lambda_{\varrho F}}
        \arrow{d}[left]{\perms(\V)×Fj}
        \descto[xshift=-0mm]{dr}{\text{\item}} & F(\perms(\V) ×
        \varrho F) \arrow{d}[right]{F(\perms(\V)×j)}
        \\
        \perms(\V) × \nu F \arrow{rr}{\perms(\V)× t} && \perms(\V) × F
        \nu F \arrow{r}{\lambda_{\nu F}} & F(\perms(\V) × \nu F)
      \end{tikzcd}
    \end{inparaenum}
  \end{equation*}
  This diagram commutes because
  \begin{inparaenum}[(1)]\item 
$j$ is a coalgebra homomorphism and
\item  $\lambda$ is natural.
  \end{inparaenum}
By finality of $\nu F$,
  $j\cdot u = a\cdot (\id_{\perms(\V)}×j)$.  As $j$ is monic, this
  means that $u$ is the restriction of the group action $a$ to
  $\varrho F$ and hence a group action.
\end{proof}

\begin{definition}\label{coalgiteration}
  For a coalgebra $c: C\to HC$ of a functor $H:\C\to\C$, we denote the
  \emph{iterated} coalgebra structure by $c^{(n)}: C\to H^nC$,
  $n\ge 0$, which is inductively defined by $c^{(0)} = \id_C$ and
    \[
        c^{(n+1)} \equiv\big(
            \begin{tikzcd}
            C \arrow{r}{c^{(n)}}
            & H^n C\arrow{r}{H^nc}
            & H^{n+1} C
            \end{tikzcd}
        \big).
    \]
    Moreover, given another functor $M$ on $\C$ and a natural
    transformation $\varphi: MH\to HM$, we define the \emph{iterated}
    transformation $\varphi^{(n)}: MH^n\to H^nM$ by
    $\varphi^{(0)} = \id$ and
    $\varphi^{(n+1)} = H^n\varphi \cdot \varphi^{(n)}H$.
  \end{definition}
  \noindent It is easy to verify that in the case where $M$ is a monad
  and $\varphi$ a distributive law of $M$ over $H$, the iterated
  transformation $\lambda^{(n)}$ is a distributive law of $M$ over
  $H^n$. These two notions of iteration interact nicely:
\begin{lemma}\label{coalgtransiteration}
  For $\varphi$ and $c$ as in Definition~\ref{coalgiteration},
  $(\varphi_C\cdot Mc)^{(n)} = \varphi^{(n)}_{C}\cdot Mc^{(n)}$.
\end{lemma}
\begin{proof}
    For $n = 0$, the equality reduces to $\id_{MC} = \id_{MC}$. For the
    induction step, we have that
    \begin{equation*}
      \begin{inparaenum}[(1)]
\begin{tikzcd}[column sep=1.2cm]
          MC \arrow{r}{Mc^{(n)}}
          \arrow[shiftarr={yshift=6mm}]{rr}{c^{(n+1)}}
          \arrow{dr}[sloped,below]{(\varphi_C\cdot Mc)^{(n)}}
          \arrow[shiftarrDR={yshift=-6mm}]{drrr}[pos=0.75,below]{(\varphi_C\cdot
            Mc)^{(n+1)}} & MH^{n}C \arrow{r}{MH^nc}
          \arrow{d}{\varphi^{(n)}_C}[left,xshift=-1mm]{\text{\item}}
          \descto{dr}{\text{\item}} & MH^{n+1}C
          \arrow{d}[left]{\varphi^{(n)}_{HC}}[right,xshift=1mm]{\text{\item}}
          \arrow{dr}[sloped,above]{\varphi^{(n+1)}_{C}}
          \\
          & H^{n}MC \arrow{r}[below]{H^nM c} & H^{n}MHC
          \arrow{r}[below]{H^n \varphi_C} & H^{n+1}MC
        \end{tikzcd}
      \end{inparaenum}
  \end{equation*}
  commutes, using \begin{inparaenum}[(1)]\item the induction hypothesis, \item naturality of
  $\varphi^{(n)}$, and \item the definition of $\varphi^{(n+1)}$.
    \end{inparaenum}
\end{proof}
\noindent In Lemma~\ref{truncationinduction} below we will establish a
coinduction principle using iterated coalgebra structures. For its
soundness proof, we use that for a finitary set endofunctor the
terminal coalgebra can be obtained by an iterative construction that
we now recall.
\begin{remark}
  \label{rem:worrell}
  \begin{enumerate}[(1)]
  \item Let $H: \Set \to \Set$ be a finitary endofunctor. The \emph{terminal sequence} of $H$ is the op-chain $(H^n1)_{n< \omega}$ with the connecting maps
  \[
    H^n !: H^{n+1} 1 \to H^n 1 \qquad \text{for every $n < \omega$.}
  \]
  Its limit $H^\omega 1$ does not in general yield the terminal
  coalgebra. However, Worrell~\cite{w05} shows that by continuing the
  terminal sequence for $\omega$ more steps, one does obtain the
  terminal coalgebra. Indeed, denote by
  $\ell_{\omega,n}: H^\omega 1 \to H^n 1$ the limit projections and
  let $H^{\omega +n} 1 = H^n(H^\omega 1)$. We define the connecting
  map $\ell_{\omega+1,\omega}: H^{\omega+1} 1 \to H^\omega 1$ as the
  unique morphism such that
  $\ell_{\omega,n} \cdot \ell_{\omega+1,\omega} = H\ell_{\omega,n}$,
  and by applying $H$ iteratively we obtain
  $\ell_{\omega + n +1, \omega + n} = H^n\ell{\omega+1,\omega}:
  H^{\omega + n + 1}1 \to H^{\omega + n} 1$.
  Worrell proves that the limit of the ensuing op-chain formed by the
  $H^{\omega + n} 1$ is the terminal coalgebra $\nu H$. Moreover, he
  shows that all connecting morphisms
  $\ell_{\omega + n + 1, \omega + n}$ are injective maps; it follows
  that $\nu H$ is actually the intersection of all $H^{\omega + n} 1$.

\item Recall that every coalgebra $c: C \to HC$ induces a
  \emph{canonical cone} $c_n: C \to H^n 1$, $n < \omega + \omega$ on
  the above op-chain defined by (transfinite) induction as follows:
  $c_0: C \to 1$ is uniquely determined, for isolated steps one has
  $c_{n+1} = Hc_n \cdot c$ and for the limit step we define $c_\omega$
  to be the unique map such that
  $\ell_{\omega,n}\cdot c_\omega = c_n$.  Note that the unique
  $H$-coalgebra morphism $c^\dagger: C \to \nu H$ can be obtained as
  the unique map such that
  $\ell_{\omega+\omega, n} \cdot c^\dagger = c_{n}$ for every
  $n < \omega + \omega$, where the maps
  $\ell_{\omega+\omega, n}: \nu H \to H^n 1$ are the limit
  projections.
\end{enumerate}
\end{remark}
\begin{lemma}\label{truncationinduction}
  Let $H: \Set\to \Set$ be a finitary endofunctor. If for
  $H$-coalgebras $(C,c)$ and $(D,d)$ there is an object $X$ with maps
  $p_1: X\to C$ and $p_2: X\to D$ such that
    \begin{equation}
    \label{coindcondition}
    \begin{tikzcd}
        X
            \arrow{r}{p_1}
            \arrow{dr}[below left]{p_2}
        & C \arrow{r}{c^{(n)}}
        & H^n C
            \arrow{dr}{H^n !}
    \\
    {}
    & D \arrow{r}{d^{(n)}}
    & H^n D \arrow{r}{H^n !}
    & H^n 1
    \end{tikzcd}
    \end{equation}
    commutes for all $n < \omega$, then $c^\dagger \cdot p_1 = d^\dagger \cdot p_2$.
\end{lemma}
\begin{proof}
  First, an easy induction shows that for the canonical cone
  $c_n: C \to H^n 1$ we have $c_n = H^n ! \cdot c^{(n)}$ for every
  $n < \omega$. Now it follows from Remark~\ref{rem:worrell} that
  elements $x \in C$ and $y \in D$ are behaviourally equivalent,
  i.e.~$c^\dagger(x) = d^\dagger(y)$, if and only if
  $c_n(x) = d_n (y)$ for all $n < \omega$. Indeed, necessity is
  obvious since $c_n = \ell_{\omega+\omega,n} \cdot c^\dagger$ (and
  similarly for $d$) and sufficiency follows from the fact that all
  $\ell_{\omega + n, \omega}$ are injective. % : using this we see that
  % $c_n(x) = d_n(y)$ for all $n < \omega$ implies the same equation for
  % all $n < \omega + \omega$

    By hypothesis we have for every $x \in X$ that 
    \[
      c_n (p_1(x)) = H^n ! \cdot c^{(n)} \cdot p_1 (x) = H^n ! \cdot d^{(n)} \cdot p_1 = d_n(p_2(x)),
    \]
    and equivalently, $c^\dagger (p_1(x)) = d^\dagger(p_2(x))$, which completes the proof.     
\end{proof}
\begin{remark}\label{betaRemark} 
  Consider the nominal sets $\bar F^n D (\varrho F)$ for $n < \omega$
  (recalling that $DX$ is $X$ equipped with the trivial nominal
  structure). We denote by
    \[
        \beta_n: \perms(\V)×F^n (\varrho F) \to F^n (\varrho F)
    \]
    the group action on $\bar F^n D(\varrho F)$. Note that for $n=0$,
    the action $\beta_0$ is trivial, i.e.~it is the projection
    \[
      \beta_0 = \outr: \perms(\V) × \varrho F \to \varrho F.
    \]
    By Assumption~\ref{lambdaassumption}, the lifting $\bar F$ is specified
    by a distributive law. Hence, we have
    $\beta_{n+1} = F\beta_n \cdot \lambda_{F^n\varrho}$. An easy
    induction thus shows that $\beta_n$ has the form
    \begin{equation}
        \label{betandef}
        \beta_n = \big(
        \perms(\V) × F^n(\varrho F)
        \xrightarrow{\lambda^{(n)}_{\varrho F}} F^n(\perms(\V)× \varrho F)
        \xrightarrow{F^n\outr} F^n(\varrho F)
        \big).
    \end{equation}
\end{remark}
\begin{lemma}
    For the $\perms(\V)$-set structure from Lemma~\ref{gsetstructure}, $t\in
    \varrho F$ is supported by
    \[
        s(t) = \bigcup_{n\ge 0} \supp(r^{(n)}(t))\quad\text{where $r^{(n)}: \varrho F \to F^n(\varrho F)$ }
    \]
    and where the support of $r^{(n)}(t)$ is taken in $\bar F^n D(\varrho F)$.
\end{lemma}
\begin{proof}
  Let $\pi \in \Fix(s(t))$. Abbreviate the coalgebra structure on
  $\perms(\V)×\varrho F$ by
  $p := \lambda_{\varrho F} \cdot (\id_{\perms(\V)}× r)$, and recall
  from Lemma~\ref{gsetstructure} that the induced algebra structure on
  $\varrho F$ is $p^\dagger$. Now consider the following diagram:
    \begin{inparaenum}[(1)]
    \[
    \begin{tikzcd}
        1
        \arrow{r}{(\pi,t)}
        \arrow{dd}[left]{t}
        \descto[xshift=-3mm]{ddr}{\text{\item}}
        &\perms(\V) × \varrho F
        \arrow{d}[left]{\perms(\V) × r^{(n)}}
        \arrow[to path=-| (\tikztotarget) \tikztonodes]{dr}[above right,pos=0.2]{p^{(n)}}
        \descto[yshift=2mm]{dr}{\text{Lemma~\ref{coalgtransiteration}}}
        &
        \\
        {}
        &\perms(\V) × F^n(\varrho F)
         \arrow{r}{\lambda^{(n)}_{\varrho F}}
         \arrow{d}[left]{\outr}
         \arrow{dr}{\beta_n\quad\text{\item}}
        &F^n(\perms(\V) × \varrho F)
        \descto[pos=0.75]{dl}{(*)}
        \arrow{d}{F^n\outr}
        \arrow[to path=-|(\tikztotarget) \tikztonodes]{dr}{F^n!}
        \descto{dr}{\text{\item}}
        \\
        \varrho F
        \arrow{r}{r^{(n)}}
        & F^n(\varrho F)
            \arrow{r}{F^n\id}
        & F^n(\varrho F)
        \arrow{r}{F^n!}
        & F^n 1
    \end{tikzcd}
    \]
    \end{inparaenum}
    Part\begin{inparaenum}[(1)]\item commutes trivially, Part~\item by
      the definition of $\beta_n$, and Part~\item by finality.
    \end{inparaenum}
    For
    \[
        (\pi,r^{(n)}(t)) \in {\perms(\V)× F^n(\varrho F)},
    \]
    Part~$(*)$ commutes as well, because $\pi$ fixes
    ${\supp(r^{(n)}(t))\subseteq s(t)}$, and thus
    $\pi \cdot r^{(n)}(t) = r^{(n)}(t)$ holds in
    $\bar F D(\varrho F)$, i.e. we have
    $\beta_n(\pi,r^{(n)}(t)) = r^{(n)}(t)$.  It follows that the
    outside of the diagram commutes, and we obtain by
    Lemma~\ref{truncationinduction} that $(\pi,t)$ and $t$ are
    identified in $\nu F$ and thus also in its subcoalgebra
    $\varrho F$. In other words, $\pi\cdot t = t$ with respect to the
    algebra structure $p^\dagger: \perms(\V)×\varrho F \to \varrho F$,
    and therefore $s(t)$ supports $t$.
\end{proof}

\begin{lemma}
  For $t\in \varrho F$, $s(t)$ is finite.
\end{lemma}
\begin{proof}
  Since $\varrho F$ is lfp in \Set, we have a finite subcoalgebra
  $j: (C,c) \to (\varrho F,r)$ containing $t$. Then define
    \[
        S = \bigcup_{x\in C} \supp(r\cdot j(x)) \subseteq \V,
    \]
    where the support is taken in $\bar FD(\varrho F)$. Clearly
    finiteness of $C$ implies that $S$ is finite.  We will now show
    that $s(t) \subseteq S$ by proving that $S$ supports
    $r^{(n)}(t) \in \bar F^nD(\varrho F)$ for every $n < \omega$. This
    follows once we show that the diagram
    \begin{equation}
    \label{supportClaim}
    \begin{tikzcd}[column sep=13mm]
        G×C \arrow{r}{G×c{}^{(n)}}
        & G×F^n C
            \arrow[yshift=2]{r}{\beta_n'}
            \arrow[yshift=-2]{r}[below]{\outr_n'}
        & F^n\varrho F
    \end{tikzcd}
    \end{equation}
    commutes, with $G = \perms(\V \setminus S)$,
    $m: G\hookrightarrow \perms(\V)$, $\beta_n'$ the restriction of
    group action $\beta_n$ on $\bar F^n D(\varrho F)$ to $G$ and $C$,
    i.e.~$\beta_n' = \beta_n\cdot (m×F^nj)$, and
    $\outr_n' = \outr \cdot (m×F^nj) = F^n j \cdot \outr$.

    To show commutation of~\eqref{supportClaim}, we proceed by
    induction in $n$. For $n=0$, \eqref{supportClaim} is clear,
    because $\beta_0 = \outr$.  For the induction step consider the
    diagram
    \[
    \begin{inparaenum}[(1)]
    \begin{tikzcd}[row sep=12mm]
        &{} \descto[pos=0.7]{d}{\text{Def.}}
        & G×F^{n+1} C
            \arrow{r}{\lambda_{F^nC}^G}
            \arrow[to path={
                        -- ([yshift=3mm]\tikztostart.north)
                        -| (\tikztotarget) \tikztonodes
                    }]{rrd}[pos=0.25]{\beta_{n+1}'}
            \descto[xshift=-4mm]{d}{\text{Naturality}}
        & F(G×F^nC)
            \arrow[xshift=2mm]{dr}{F\beta_{n}'}
            \descto{d}{F(\text{IH})}
            \descto{r}{\text{\item}}
        &{}
        \\
        G×C \rar{G×c}
            \arrow[bend left=10]{urr}[sloped,above]{G×c^{(n+1)}}
            \arrow{dr}[sloped,below]{G×c^{(n+1)}}
            \descto[xshift=8mm]{dr}{\text{Def.}}
        & G× FC \rar{\lambda_C^G}
            \arrow[bend left=10]{ur}[sloped,above]{G×Fc^{(n)}}
            \arrow{d}{G×Fc^{(n)}}
            \arrow{drr}[sloped,above,pos=0.6]{\beta_1'}
            \arrow[bend right=20]{drr}[sloped,above,pos=0.7]{\outr_1'}
        & F(G×C) \arrow[xshift={-1mm}]{ur}[sloped,above]{F(G×c^{(n)})\ }
                \arrow{r}[above]{F(G×c^{(n)})}
                \arrow{dr}[sloped,above]{F\outr_0'}
                \descto[pos=0.15]{d}{\text{\item}}
                \descto[pos=0.60]{d}{\text{\item}}
        & F(G×F^nC) \arrow{r}{F\outr_n'}
                \descto[pos=0.5]{d}{\text{\item}}
        & F^{n+1}\varrho F
        \\
        &|[yshift=-5mm]|G× F^{n+1}C
            \arrow[to path={-|(\tikztotarget)\tikztonodes}]{urrr}[below,pos=0.3]{\outr_{n+1}'}
        &
        {} \descto[pos=0.0,yshift=-2mm]{u}{\text{\item}}
        & F\varrho F \arrow{ur}[above]{Fr^{(n)}\ \ }
    \end{tikzcd}
    \end{inparaenum}
    \]
    Most parts commute by the definitions of $\beta_n$, $\beta_n'$, and
    $c^{(n+1)}$, respectively. For the remaining parts:
    \begin{enumerate}[(1)]
    \item is the commutative diagram below:
    \[
    \hspace{-6mm}
    \begin{tikzcd}[column sep =0.9cm]
        G×F^{n+1}C
            \arrow{r}[yshift=1.5mm]{m×F^{n+1}C}
            \arrow{d}[left]{\lambda^G_{F^nC}}
            \descto{dr}{\text{Assumption~\ref{lambdaassumption}}}
            \arrow[
                to path={
                        -- ([yshift=3mm]\tikztostart.north)
                        -|(\tikztotarget)\tikztonodes}
                ]{rrrd}[pos=0.25]{\beta_{n+1}'}
        & \perms(\V)×F^{n+1}C
            \arrow{d}{\lambda_{F^nC}}
            \arrow{r}[yshift=1.5mm]{\perms(\V)×F^{n+1}j}
            \descto{dr}{\text{Naturality of }\lambda}
        & \perms(\V) × F^{n+1}\varrho F
            \arrow{d}[left]{\lambda_{\varrho F}}
            \arrow[xshift=3mm,yshift=1mm]{dr}{\beta_{n+1}}
            \descto[xshift=-6mm]{dr}{\text{Remark~\ref{betaRemark}}}
        \\
        F(G×F^nC)
            \arrow{r}[yshift=-1.5mm,below]{F(m×F^nC)}
            \arrow[shiftarr={yshift=-7mm}]{rrr}[below]{F\beta_n'}
        & F(\perms(\V)×F^nC)
            \arrow{r}[yshift=-1.5mm,below]{F(\perms(\V)×F^nj)}
        &F(\perms(\V)×F^n\varrho F)
            \arrow{r}[yshift=-1.5mm,below]{F^n\beta_n}
        & F^{n+1}\varrho F
    \end{tikzcd}
    \]
    \item is just the previous item for $n=0$ using that $\outr_0' = \beta_0'$.
    \item commutes, i.e.~$\beta_1'= \outr_1'$, because $G$ is defined to consist of those
    permutations that fix every element in $S$ and therefore fix every element in
    the image of $r\cdot j$.
    \item commutes because we can remove $F$ and then prove $\outr_n'\cdot (\id_G×c^{(n)}) =
    r^{(n)}\cdot \outr_0'$ by induction. For $n=0$, the desired equation obviously holds. For the induction step 
    consider the commutative diagram below:
    \[
    \begin{tikzcd}[column sep=1.1cm]
        G×C
            \rar{G×c^{(n)}}
            \arrow[shiftarr={yshift=8mm}]{rr}{G×c^{(n+1)}}
            \arrow{d}[left]{\outr}
            \descto{ddr}{\,\text{(IH)}}
            \arrow[shiftarr={xshift=-8mm}]{dd}[left]{\outr_0'}
        & G×F^nC
            \arrow{d}[left]{\outr}
            \descto{dr}{\text{Naturality}}
            \arrow{r}{G×F^nc}
        & G×F^{n+1}C
            \arrow{d}{\outr}
            \arrow[shiftarr={xshift=12mm}]{dd}{\outr_{n+1}'}
        \\
        C
            \arrow{d}[left]{j}
        & F^n C
            \arrow{d}[left]{F^nj}
            \arrow{r}{F^nc}
            \descto{dr}{\,\text{Coalgebra Hom.}}
        & F^{n+1} C
            \arrow{d}{F^{n+1}j}
        \\
        \varrho F
            \arrow{r}{r^{(n)}}
            \arrow[shiftarr={yshift=-5mm}]{rr}[below]{r^{(n+1)}}
        & F^n\varrho F
            \arrow{r}{F^n r}
        & F^{n+1}\varrho F
    \end{tikzcd}
    \]
    \item commutes by a similar induction proof as the previous item, but starting with $n=1$.
    \qedhere
    \end{enumerate}
\end{proof}
\noindent The two previous lemmas combined imply that $\varrho F$ is a
nominal set; by Lemma~\ref{gsetstructure}, $(\varrho F, r)$ is a
$\bar F$-coalgebra, and by Lemma~\ref{setLfp2NomLfp} this coalgebra is
lfp.
\begin{theorem}\label{thm:lift}
    The lifted coalgebra $(\varrho F, r)$ is the rational fixpoint of $\bar F$.
\end{theorem}
\begin{proof}
  It remains only to show that for every $\bar F$-coalgebra $(C,c)$ in
  $\Nom$ with orbit-finite carrier there exists a unique coalgebra
  homomorphism from $C$ to $\varrho F$. So let $(C,c)$ be an
  orbit-finite $\bar F$ coalgebra. By Lemma~\ref{nomLfp2SetLfp}, $(C,c)$
  is an lfp $F$-coalgebra, and thus induces a unique $F$-coalgebra
  homomorphism $h: (C,c)\to (\varrho F,r)$ (in $\Set$). This
  homomorphism $h: C\to \varrho F$ is equivariant; to see this, recall
  first that the final $F$-coalgebra $(\nu F, t)$ lifts to the final
  $\hat F$-coalgebra, where $\hat F$ is the lifting of $F$ to $\perms(\V)$-sets
  induced by the distributive law. Let $h': (C, c) \to (\nu F, t)$ be the unique
  homomorphism into the final $\hat F$-coalgebra, which is an
  equivariant map. Recall further that $(\varrho F, r)$ is a subcoalgebra
  of $(\nu F, t)$ via $j: \varrho F \to \nu F$, say. Note that the
  group action on $\varrho F$ has been defined as the restriction of that on
  $\nu F$, and so $j$ is equivariant. Then clearly $j \cdot
  h= h'$ by finality of $\nu F$; and since $j$ is equivariant and monic it
  follows that $h$ is equivariant.
  \end{proof}

\begin{example}\label{ex:rats}
  \begin{enumerate}[(1)]
  \item For all canonical liftings $\bar F$ (e.g.~the ones we mentioned in Example~\ref{setFunctorList}), the rational fixpoint $\varrho \bar F$ is the rational fixpoint $\varrho F$ equipped with the discrete action (cf.~Example~\ref{ex:coalgs}). Note that for the cyclic shift functor~$\mathcal Z$ on $\Set$, the final coalgebra consists of all finitely branching trees where the order of the children of any vertex is taken modulo cyclic shifting, and the rational fixpoint is given by all rational such trees; this follows from the results of~\cite{am06}.
  \item Recall that the final coalgebra for $ LX = \V + X \times X + \V \times X$ on $\Set$ is carried by the set of all $\lambda$-trees and the rational fixpoint by the set of all rational $\lambda$-trees. It follows that the rational fixpoint of the (non-canonical but localizable) lifting~$\bar L$ where $\V$ is equipped with the standard action is carried by the same set with the nominal structure given according to Lemma~\ref{gsetstructure}; this action applies the standard action of $\V$ to the labels of the leaves of $\lambda$-trees. 
  \item For the functor $\Bag(-) + \V$ on $\Set$, the final coalgebra is carried by all unordered trees some of whose leaves are labelled in $\V$. The rational fixpoint is then carried by the set of all rational such trees. To obtain the rational fixpoint of the non-canonical lifting $\bar{\Bag}(-) + \V$, one equips this set of trees with the action that applies the standard action of $\V$ on the labels of leaves. This follows once again from Lemma~\ref{gsetstructure} and Theorem~\ref{thm:lift}. A similar description can be given for the functor $\bar{\mathcal Z} + \V$; we obtain rational trees some of whose leaves are labelled in $\V$ and where the order of the children of a vertex is only determined up to cyclic shift. 
  \end{enumerate}
\end{example}

\section{Quotients of \Nom-Functors}\label{sec:quotients}

We next consider quotient functors on $\Nom$. For the rest of this
section we assume a finitary functor $H: \Nom \to \Nom$ that is a
\emph{quotient} of a finitary functor $F: \Nom \to \Nom$, i.e.~we have
a natural transformation $q: F \twoheadrightarrow H$ with surjective components. We
present a sufficient condition on coalgebras for $F$ and $H$ that
ensures that the rational fixpoint $\varrho H$ is a quotient of the
rational fixpoint~$\varrho F$. We then introduce a simple, if ad-hoc,
condition on $F$ that ensures that the mentioned sufficient condition
is satisfied for all quotients $H$ of $F$. Combining this result with
the ones from the previous section, we obtain a description of the
rational fixpoint of endofunctors $H$ on $\Nom$ arising from binding
signatures and exponentiation.

\begin{definition}\label{coalgquotient}
  An $H$-coalgebra $(C,c)$ is a \emph{quotient} of an
  $F$-coalgebra $(A,a)$ if there is a surjective $H$-coalgebra
  homomorphism $h: (A,q_A\cdot a) \to (C,c)$.
\end{definition}
\begin{theorem}
  \label{thm:quot}
  Suppose that every orbit-finite $H$-coalgebra is a quotient of an
  orbit-finite $F$-coalgebra. Then the rational fixpoint of $H$ is a
  quotient of the rational fixpoint of $F$.
\end{theorem}
\begin{proof}
  Let $(\varrho H, r^H)$ and $(\varrho F, r^F)$ be the rational
  fixpoints of $H$ and $F$, respectively.  First we argue that the
  $H$-coalgebra
  \begin{equation}\label{eq:lfp}
    \varrho F \xrightarrow{r^F} F(\varrho F) \xrightarrow{q_{\varrho F}} H(\varrho F)
  \end{equation}
  is lfp. To this end note that the object assignment that maps an
  $F$-coalgebra $(A,a)$ to the $H$-coalgebra $(A, q_A \cdot a)$
  extends to a finitary functor $\Coalg\, F \to \Coalg\, H$ that
  preserves orbit-finite coalgebras. So since $(\varrho F, r^F)$ is
  the filtered colimit of all orbit-finite $F$-coalgebras $(A, a)$,
  the above $H$-coalgebra~\eqref{eq:lfp} on $\varrho F$ is the
  filtered colimit of all orbit-finite $H$-coalgebras of the form
  $(A, q_A\cdot a)$, whence~\eqref{eq:lfp} is an lfp coalgebra for
  $H$.

  Now we obtain a unique $H$-coalgebra homomorphism
  $p: (\varrho F, q_{\varrho F}\cdot r^F) \to (\varrho H, r^H)$ by the
  finality of the latter coalgebra. It remains to show that $p$ is
  surjective. To this end, let $(C,c)$ be an orbit-finite
  $H$-coalgebra. By assumption, $(C,c)$ is a quotient of some
  orbit-finite $F$-coalgebra $(A, a)$, i.e.\ we have a diagram
  \begin{equation}\label{diag:quotient}
    \begin{tikzcd}
        (A,q_A\cdot a) \arrow{r}{a^\dagger}
            \arrow[->>]{d}[left]{h}
        &
        (\varrho F, q_A\cdot r^F)
            \arrow{d}{p}
        \\
        (C,c) \arrow{r}{c^\dagger}
        &
        (\varrho H,r^H)
    \end{tikzcd}
  \end{equation}
  Now the sink consisting of all these $c^\dagger \cdot h$, where
  $(C,c)$ ranges over all orbit-finite $H$-coalgebras, is jointly
  surjective since the $c^\dagger$ are jointly surjective. Thus, it
  follows that $p$ is surjective by commutativity of the
  diagrams~\eqref{diag:quotient}.
\end{proof}

\begin{remark}
  It follows from the previous theorem that $\varrho H$ is the image
  of ${(\varrho F, q_{\varrho F}\cdot r^F)}$ in the final
  $H$-coalgebra. In fact, $\varrho H$ is a subcoalgebra of $\nu H$ via
  the injective $H$-coalgebra homomorphism $m: \varrho H \to \nu H$,
  say. Then $m \cdot p$ is the image-factorization of the
  unique $H$-coalgebra homormorphism from~\eqref{eq:lfp} to $\nu H$.
\end{remark}
\noindent We next introduce the announced condition on $F$ that
ensures satisfaction of the assumptions of Theorem~\ref{thm:quot}.
\begin{definition}\label{def:sstr}
  For nominal sets $X$ and $Y$, we define the nominal subset
    \[
        X<Y = \{ (x,y) \in X×Y \mid \supp(x) \subseteq \supp(y) \}
    \] 
    of $X\times Y$ (this subset is clearly equivariant, so $X<Y$ is
    indeed a nominal set).  A \emph{sub-strength} of $F$ is a family
    of a equivariant maps
    \[
      s_{X,Y}:  FX < Y \to  F(X<Y),
    \]
  indexed by nominal sets $X,Y$ (but not necessarily natural in
  $X,Y$) such that
  \begin{equation}
    \label{diag:substrength}
    \begin{tikzcd}
       FX < Y \arrow{r}{s_{X,Y}} \arrow{dr}[below left]{ \outl} &
       F(X< Y) \arrow{d}{ F\outl} \\ {} &  FX
    \end{tikzcd}
  \end{equation}
  where $\outl: X < Y \to X$ denotes the obvious projection map.
\end{definition}
\begin{example} \label{ex:noSubStrength}
  Not every functor has a sub-strength. The finitary functor $DU$ is a
  lifting of $\Id_\Set$ (in fact, it is a mono-preserving and
  localizable lifting) but has no sub-strength, because for $X=\V$,
  $Y=1$, the nominal set $DU\V < 1$ is just $DU\V$ whereas
  $DU(\V < 1) = DU\emptyset = \emptyset$. Hence, there is no map
  $DU\V<1\to DU(\V<1)$ at all.
\end{example}
\noindent
For the rest of this section, we assume that $F$ has a
sub-strength $s_X$. Moreover, we fix an orbit-finite coalgebra
$c: C\to HC$. We will show that $c$ is a quotient of an $F$-coalgebra
in the sense of Definition~\ref{coalgquotient}, thus showing that the
rational fixpoint of $H$ is a quotient of that of $ F$, by
Theorem~\ref{thm:quot}.

We put
\begin{equation}
    \label{def:B}
    B = \max_{x\in C}|\supp(x)|
       \ +\  \max_{x\in C}\min_{\substack{y\in  FC\\ q_C(y) = c(x)}}
            |\supp(y)|.
\end{equation}
Intuitively, $B$ is a bound on the total number of free and bound
variables in any element $c(x)$. First observe that $B$ exists because the numbers numbers $|\supp(x)|$ and
$\min_{y\in  FC,q_C(y) = c(x)} |\supp(y)|$ are constant on every
orbit of $C$; for the former simply apply Lemma~\ref{orbitsamesupp}, and for the latter suppose that $x = \pi \cdot x'$ and let $y$ and $y'$, respectively, assume the above minimum. Then 
\[
q(y) = c(x) = c( \pi \cdot x') = \pi \cdot c(x) = \pi \cdot q(y') = q(\pi \cdot y')
\]
using equivariance of $c$ and $q$ and therefore $|\supp(y)| \leq |\supp(\pi \cdot y')| = |\supp(y')|$ by minimality and Lemma~\ref{orbitsamesupp}. Similarly $|\supp(y')|\leq |\supp(y)|$ by starting from $x' = \pi^{-1} \cdot x$. 

Next we define $W\subseteq \V^B$ to be the nominal
set of tuples of $B$ distinct atoms. Thus, for every $w \in W$, $|\supp(w)| = B$. 

Note that $W$ has only one orbit, in particular is orbit-finite.
Hence $C×W$ and thus also its subobject $C<W$ are
orbit-finite.
We will use $C<W$ as the carrier of the orbit-finite $F$-coalgebra we
aim to construct.

\begin{lemma}\label{lem:outl-epi}
    The projection $\outl: C<W \to C$ is an epimorphism.
\end{lemma}
\begin{proof}
  For $x\in C$, there is $w\in W$ with $\supp(x)\subseteq \supp(w)$,
  because $|\supp(x)|\le B$. So $(x,w)\in C<W$ and $\outl(x,w) = x$.
\end{proof}
\noindent We recall the notion of strong nominal set:
\begin{definition}\cite{Tzevelekos07}
  An element $x$ of a nominal set $X$ is \emph{strongly supported} if
  $\fix(x)\subseteq\Fix(\supp(x))$ (so $\fix(x)=\Fix(\supp(x))$). A
  nominal set is \emph{strong} if all its elements are strongly
  supported.
\end{definition}

\begin{example}
\begin{enumerate}[(1)]
\item The nominal set $W$ is strong: for $a=(a_1,\ldots,a_B)\in W$,
  the equality $\pi \cdot a = a$ implies that $\pi(a_i) = a_i$ for all
  $i$, i.e.\ $\pi\in\Fix(\supp(a))$.

\item The nominal set of unordered pairs of atoms fails to be strong,
  because $(a\ b)\cdot \{a,b\} = \{a,b\}$.
\end{enumerate}
\end{example}
\noindent 
From the first example, the following is immediate.
\begin{lemma}
  The nominal set $C < W$ is strong. 
\end{lemma}
\noindent Strong nominal sets are of interest due to the following
extension property (mentioned already in~\cite{KurzEA10}):
\begin{proposition}\label{orderedOrbitProperty}
  Let $X$ be a strong nominal set, and let $O$ be a subset of $X$
  containing precisely one element per orbit of $X$. Let $Y$ be a
  nominal set, and let $f_0: O \to Y$ be a map such that
  $\supp(f_0(x)) \subseteq \supp(x)$ for all $x\in O$. Then $f_0$
  extends uniquely to an equivariant map $X\to Y$.
\end{proposition}
\begin{proof}
  Uniqueness is clear. To show existence, define $f:X\to Y$ by
  $f(\pi\cdot x) = \pi\cdot f_0(x)$ for $x\in O$. We have to show
  well-definedness, so let $\pi' \cdot x = \pi \cdot x$. Since $X$ is
  strong and $\supp(f_0(x))\subseteq\supp(x)$, we then have
  $\pi^{-1}\pi'\in\fix(x)\subseteq\Fix(\supp(x))\subseteq\Fix(\supp(f_0(x)))\subseteq\fix(f_0(x))$,
  so $\pi'\cdot f_0(x) = \pi\cdot f_0(x)$.  Equivariance of $f$ is
  immediate from the definition.
\end{proof}
\noindent This property is used in the construction of a part of our
target coalgebra. In the construction, it is an essential observation that if an
equivariant map drops certain atoms, then we can rename the atoms without
changing the value:
\begin{lemma} \label{changeOfSupp}
    Consider an equivariant map $e: X\to Y$ and $x\in X$. Then for
    any $S\in \Potf(\V)$ with $\supp(e(x))\subseteq S$ and $|\supp(x)| \le |S|$,
    there is some $\pi\in\perms(\V)$ with $\supp(\pi\cdot x) \subseteq S$ and
    $e(\pi\cdot x) = e(x)$.
\end{lemma}
\begin{proof}
    Put $Y = \supp(x) \setminus \supp(e(x))$ and $N=S\setminus\supp(e(x))$. Then
    $|Y| \le |N|$. Pick some injection $\pi': Y\setminus N \rightarrowtail
    N\setminus Y$ and extend it to a finite permutation on $\V$ by
    \[
        \pi(a) = \begin{cases}
            \pi'(a) & \text{if }a\in Y\setminus N \\
            \pi'^{-1}(a) & \text{if }a\in \Im(\pi') \\
            a & \text{otherwise}
        \end{cases}
    \]
    where $\Im(\pi') \subseteq N\setminus Y$ denotes the image of
    $\pi'$. This definition implies $\pi[Y] \subseteq N$ and
    $\pi\cdot e(x) = e(x)$ since $\pi$ fixes every $a \not\in Y$ and therefore $\pi
    \in \Fix(\supp(e(x)))$. Hence,
    \begin{align*}
        \supp(\pi\cdot x)
        &= \pi\cdot \supp(x)
        \subseteq \pi\cdot (Y\cup \supp(e(x)))
        = (\pi\cdot Y) \cup \supp(\pi\cdot e(x))
        \\ &
        = \pi[Y] \cup \supp(e(x))
        \subseteq N \cup \supp(e(x))
        = S.
        \tag*{\qedhere}
    \end{align*}
\end{proof}
\begin{lemma}\label{equivConstruction}
  There is an equivariant map $f: C < W \to FC$ such that
    \[
    \begin{tikzcd}
        C< W
            \arrow{r}{f}
            \arrow[->]{d}[left]{\outl}
        & FC
            \arrow[->]{d}{q_C}
        \\
        C \arrow{r}{c}
        & HC
    \end{tikzcd}
    \]
    commutes.
\end{lemma}
\begin{proof}
  Pick a subset $O=\{(x_1,w_1),\ldots,(x_n,w_n)\} \subseteq C <W$
  containing precisely one element from each of the $n$ orbits of
  $C<W$.  We have for each $i$ some $y_i \in FC$ such that
  $q_C(y_i) = c(x_i)$, and by \eqref{def:B},
  $|\supp(y_i)| \le B = |\supp(w_i)| $; in addition,
  $\supp(q_C(y_i))\subseteq \supp(x_i)\subseteq \supp(w_i)$.  By
  Lemma~\ref{changeOfSupp} applied to $q_C$, $y_i$ and $S = \supp(w_i)$, there is some
  $\sigma_i$ such that $q_C(\sigma_i\cdot y_i) = q_C(y_i) = c(x_i)$
  and
  $\supp(\sigma_i\cdot y_i) \subseteq \supp(w_i) = \supp(x_i,w_i)$.

    Now define $f_0: O\to FC$, $f_0(x_i,w_i) = \sigma_i\cdot
    y_i$.
    Then $\supp(f_0(x_i,w_i)) \subseteq \supp(x_i,w_i)$. By
    Proposition~\ref{orderedOrbitProperty}, $f_0$ extends uniquely to an
    equivariant map $f: C<W\to F C$, and we have
    \begin{equation}
        \label{fTempProperty}
        q_C(f(x_j,w_j))
        = q_C(f_0(x_j,w_j))
        = q_C(\sigma_j\cdot y_j)
        = c(x_j)
        = c\cdot \outl(x_j,w_j)
    \end{equation}
    for all $1\le j\le n$.  This equality extends to all elements of
    $C< W$ by equivariance: any $p \in C< W$ has the form
    $p = \pi\cdot (x_i,w_i)$, and thus multiplying
    \eqref{fTempProperty} by $\pi$ yields
    $q_C(f(p)) = c\cdot \outl(p)$.
\end{proof}
\noindent In combination with the sub-strength, the map $f$ now
induces the required $F$-coalgebra:
\begin{lemma}
  The $H$-coalgebra $(C,c)$ is, via $\outl$, a quotient of the
  orbit-finite $F$-coalgebra
    \[
    \begin{tikzcd}
        C < W
        \arrow{r}{\bar f}
        &
        F C < W
        \arrow{r}{s_{C,W}}
        &
        F (C < W)
    \end{tikzcd}
    \quad\text{where }\bar f(x,w) = (f(x),w).
    \]
\end{lemma}
\begin{proof}
    The map $\bar f$ is equivariant, and $\bar f(x,w) \in FC<W$ because $f$ is
    equivariant. Moreover, the diagram below commutes:
    \[
    \begin{tikzcd}[column sep=1.2cm]
        C < W
        \arrow{r}{\bar f}
        \descto{rdr}{\text{Def. }\bar f}
        \arrow{dd}[left]{\outl}
        \arrow[bend right=15]{drr}[below]{f}
        &
        FC <W
        \arrow{r}{s_{C,W}}%[below,yshift=-1mm,pos=.7]{\text{\eqref{diag:substrength}}}
        \arrow{dr}[sloped,below]{\outl}
        &
        F(C<W)
        \arrow{r}{q_{C< W}}
        \arrow{d}[right]{F\outl}
        \descto[pos=.1]{dl}{\text{\eqref{diag:substrength}}}
        & 
        H(C < W)
        \ar{dd}[right]{H\outl}
        \\
        {}
        &{}& FC
        \descto[yshift=3mm,xshift=1mm]{r}{\text{Naturality}}
        \descto{dll}{\text{Lemma~\ref{equivConstruction}}}
        \arrow{dr}{q_C}
        &
        {}
        \\
        C
        \arrow{rrr}{c}
        &&& HC
    \end{tikzcd}
    \]
    Thus, $\outl: C< W\to C$ is an $H$-coalgebra homomorphism,
    and surjective by Lemma~\ref{lem:outl-epi}.
\end{proof}
\noindent
From Theorem~\ref{thm:quot} we now obtain:
\begin{corollary}\label{cor:rho-quotient}
  If $F:\Nom\to\Nom$ is finitary and has a sub-strength, and $H$ is a
  quotient of $F$, then the rational fixpoint $\varrho H$ is a
  quotient of the rational fixpoint~$\varrho F$.
\end{corollary}

\begin{example}
  Having a sub-strength is not a necessary condition for a quotient
  $F\to H$ to satisfy the requirements of Theorem~\ref{thm:quot}. Recall
  from Example~\ref{ex:noSubStrength} that the functor $F=DU$ has no
  sub-strength. Take $q: DU\twoheadrightarrow H$ to be any quotient
  (e.g.~$HX=1$). Since $q_X:DUX\to HX$ is a surjective equivariant map
  and equivariant maps do no increase the support, $HX$ is a discrete
  nominal set for all $X$. It follows that every splitting
  $s_X: HX \rightarrowtail FX$ of $q_X$ is an equivariant map. Thus,
  every $H$-coalgebra $x: X\to HX$ is trivially a quotient of the
  $DU$-coalgebra $s_X\cdot x: X\to DUX$, i.e.\ the quotient $DU\to H$
  satisfies the assumptions of Theorem~\ref{thm:quot}.
\end{example}
\noindent However, a sub-strength does exist in many relevant
examples:
\begin{lemma}\label{lem:substrengths}
  \begin{enumerate}[(1)]
  \item Every constant functor has a sub-strength.
  \item The identity functor has a sub-strength.
  \item The class of functors having a sub-strength is closed under
    finite products, arbitrary coproducts, and functor composition.
  \end{enumerate}
\end{lemma}
\begin{proof}
  We give definitions of the sub-strength in all cases; commutation of \eqref{diag:substrength} is obvious throughout.
  \begin{enumerate}[(1)]%[(1)]
  \item If $K$ is constant, then we have $s_{X,Y} = \outl: KX < Y \to K(X<Y)$.
  \item Trivial. % ly $s_{X,Y} = \id$.
  \item
    \begin{enumerate}\item 
      For $G$ and $H$ having sub-strengths $s^G_{X,Y}$ and
      $s^H_{X,Y}$, respectively, we define
      \[
        f: (GX×HX) < Y \to (GX < Y) × (HX < Y)
      \]
      by
      $f(x,y,w) = \big((x,w),(y,w)\big)$, which is well-typed because
      $\supp(x) \subseteq \supp(x,y) \subseteq w$ (and analogously for
      $y$). We then obtain a sub-strength $s_{X,Y}$ for $G\times H$ as
      $s_{X,Y}=s^G_{X,Y}\times s^H_{X,Y}\circ f$.
    \item For each $G_i$ having a sub-strength $s^i_{X,Y}$, we define
      \[ 
        f: \big(\coprod_{i\in I} G_iX\big) < Y \to \coprod_{i\in I}(G_iX < Y) 
      \]
      by $f(\inj_i x, w) = \inj_i (x,w)$, again noting that
      $\supp(x) = \supp(\inj_i x) \subseteq \supp(w)$.  We then obtain
      a sub-strength $s_{X,Y}$ for $\coprod G_i$ as
      $s_{X,Y}=(\coprod s^i_{X,Y})\circ f$.
    \item Given sub-strengths $s^F_{X,Y}: FX < Y \to F(X < Y)$ and
      $s^G_{X,Y}: GX < Y \to G(X <Y)$, the desired sub-strength for the
      composite $GF$ is
      \[
        GFX < Y \xrightarrow{s^G_{FX,Y}} G(FX < Y) \xrightarrow{Gs^F_{X,Y}} GF(X < Y).  \qedhere
      \]
    \end{enumerate}
    \end{enumerate}
\end{proof}

\begin{notation}\label{not:neq}
  We denote by $X^{n\neq}$ the subset of $X^n$ consisting of all
  $n$-tuples with pairwise distinct components.
\end{notation}

\begin{proposition} \label{orbitFiniteEpiPreservation}
    If a functor $F:\Nom\to\Nom$ has a sub-strength, then it preserves
    epimorphisms with orbit-finite codomain.
\end{proposition}
\begin{proof}
    Take $e: X\twoheadrightarrow Y$ and suppose that $Y$ is orbit-finite.
    Define
    \[
        m = \max_{y\in Y} \min_{\substack{x\in X\\ e(x) = y}} |\supp(x)|
        \quad\text{and}\quad
        Z = \coprod_{k \ge m} \V^{k\neq}.
    \]
    The maximum $m$ exists, because $Y$ is orbit-finite and because
    for any two elements $y$, $y'$ of the same orbit,
    $|\supp(y)| = |\supp(y')|$ by Lemma~\ref{orbitsamesupp}. Pick a
    subset $O \subseteq Y < Z$ containing precisely one representative
    $(x_i,z_i)$ of each orbit of $Y < Z$. We have for every $i$ some
    $x_i$ with $e(x_i) = y_i$ and
    $|\supp(x_i)| \le m \le |\supp(z_i)|$. By Lemma~\ref{changeOfSupp}
    applied to $e$, $y_i$ and $S=\supp(z_i)$ we can assume
    w.l.o.g.~that $\supp(x_i) \subseteq \supp(z_i)$.

    Now define $c_0: O \to X < Z$ by $c_0(y_i,z_i) = (x_i, z_i)$. By Proposition~\ref{orderedOrbitProperty} we obtain a unique equivariant extension $c: Y < Z  \to X < Z$, and for every $\pi \in \perms(\V)$ we have 
    \[
        (e<\id_Z)(c(\pi\cdot y_i,\pi\cdot z_i))
        = \pi\cdot (e<\id_Z)(x_i,z_i)
        = \pi\cdot (y_i,z_i).
    \]
    This implies that $(e<\id_Z): X<Z \to Y<Z$ is a split epimorphism and thus
    preserved by $F$. 

    Next consider the commuting diagram
    \[
    \begin{tikzcd}
        F(X<Z)
        \descto[anchor=center,xshift=-3mm]{dr}{\text{\parbox{4cm}{\centering\scriptsize Naturality\\ of $\outl$}}}
        \arrow[->>]{r}{F(e<Z)}
        \arrow{d}[left]{F\outl}
        & F(Y<Z)
        \arrow{d}[left]{F\outl}
        & FY < Z
        \arrow{l}[above]{s_{Y,Z}}
        \arrow[xshift=1mm,yshift=-1mm,bend left=10]{dl}[sloped,below]{\outl}[above
        left]{\text{\eqref{diag:substrength}}}
        \\
        FX
        \arrow{r}[below]{Fe}
        & FY
    \end{tikzcd}
    \]
    For any $t \in FY$, there is some $z\in Z$ with
    $\supp(t) \subseteq \supp(z)$ since for every subset $S$ of $\V$
    of cardinality of least $m$ there exist elements in $Z$ whose
    support is $S$. Since $\outl: FY < Z\to FY$ is epimorphic, so is
    $F\outl: F(Y < Z) \to FY$. Hence, $F\outl\cdot F(e<Z)$ is
    epimorphic, thus so is $Fe$.
\end{proof}
\noindent For epimorphisms with non-orbit-finite codomain,
preservation by functors having a sub-strength may fail:
\begin{example}
    The functor $FX=X^\omega$ of finitely supported sequences has a sub-strength
    \[
        s_{X,Y}\big( (a_k)_{k\in \N}, y \big) = \big((a_k,y)\big)_{k\in \N}
    \]
    because
    $\supp(a_k)\subseteq \supp\big((a_k)_{k\in\N})\big) \subseteq
    \supp(y)$.

    However, $F$ does not preserve all epimorphisms. To see this consider the
    discrete nominal set $\N$ of natural numbers and the equivariant surjection 
    \[
        e: \Potf \V \twoheadrightarrow \N,
        \quad e(W) = |W|.
    \]
    The image of $Fe$ in $\N^\omega$ contains only bounded sequences:
    for any finitely supported sequence $s\in \Potf \V^\omega$, the
    sequence $Fe(s)$ is bounded by $|\supp(s)|$. Since $\N$ is
    discrete, every sequence in $\N^\omega$ has finite (namely, empty)
    support; this shows that $Fe$ is not surjective.
\end{example}
\noindent 
Note that $FX = X^\omega$ is not finitary (see
Proposition~\ref{orbitFiniteNecessary}). In fact, for finitary functors we
have the following
\begin{proposition}
    Finitary \Nom-functors with a sub-strength preserve epimorphisms.
\end{proposition}
\noindent This easily follows from
Proposition~\ref{orbitFiniteEpiPreservation} since every epimorphism in $\Nom$
is the filtered colimit of epimorphisms with orbit-finite domain and
codomain (see Proposition~\ref{finitaryEpiPreservation} in the appendix).

Preservation of epimorphisms is convenient because quotients of
epimorphism-preserving functors are closed under composition:
\begin{lemma}\label{lem:quot}
  Let $q: F \epito H$ and $q': F' \epito H'$ be quotients of functors
  on $\Nom$. If $F$ preserves epis then $HH'$ is a quotient of $FF'$
  via
  \[
    FF' \xrightarrow{Fq'} FH' \xrightarrow{qH'} HH',
  \]
\end{lemma}
\begin{proof}
  Recall that we have defined quotients as natural transformations
  that are pointwise epi.
\end{proof}
\noindent This extends the closure properties of the class of functors
with a sub-strength (Lemma~\ref{lem:substrengths}) to the class of
quotients of finitary functors having a sub-strength:
\begin{corollary}
  The class of quotients of finitary \Nom-functors that have a
  sub-strength is closed under coproducts, finite products,
  composition, and quotients.
\end{corollary}
\begin{proof}
  Closedness under coproducts, finite products, and quotients is
  trivial. For composition, recall that finitary functors that have a
  sub-strength preserve epimorphisms, and apply Lemma~\ref{lem:quot}.
\end{proof}

\section{Applications}

\subsection{Binding Signatures}\label{sec:binding}

One can describe various flavours of (possibly) infinite terms with variable binding operators (such as infinite $\lambda$-terms or process terms of the $\pi$-calculus) as the inhabitants of final coalgebras of so
called \emph{binding signatures}, see \cite[Definition
5.8]{nomcoalgdata}. We refrain from defining the corresponding
\emph{binding signatures} explicitly here, focusing instead on the
functor representation. The latter is given by a generalization of the
class of polynomial functors:
\begin{definition}
  The class of \emph{binding functors} is the smallest class of
  functors on $\Nom$ that contains the identity functor and all
  constant functors and is closed under all coproducts, binary
  products, and left composition with the abstraction functor
  $[\V](-)$\,. The \emph{raw functor} of a binding functor is
  the polynomial functor obtained by replacing all occurrences of
  $[\V](-)$\, with $\V\times(-)$\,. (Strictly speaking
  this requires an explicit distinction between a syntax and a
  semantics for binding functors; we refrain from elaborating this
  distinction to avoid overformalization.)
\end{definition}
\begin{lemma}
  Every binding functor is a quotient of its raw functor.
\end{lemma}
\noindent By Lemma~\ref{lem:substrengths} and
Corollary~\ref{cor:rho-quotient}, we have in particular that for every
quotient~$H$ of a polynomial functor $F$, the rational fixpoint
$\varrho H$ is a quotient of~$\varrho F$. By the previous lemma, this
applies in particular in the situation where $H$ is a binding functor
and $F$ is its raw functor. One concrete instance is the main result of \cite{mw15}:
\begin{example}
  For $FX = \V  + \V×X + X×X$ and $HX = \V + [\V]X + X×X$ we already saw that the
  rational fixpoint $\varrho F$ is formed by all rational $\lambda$-trees (see
  Example~\ref{ex:rats}(2)). Furthermore, we know that $\varrho H$ is a subcoalgebra
  of
  the final $H$-coalgebra, and the latter consists of the
  $\alpha$-equivalence classes of $\lambda$-trees with finitely many
  free variables~\cite{nomcoalgdata}. But now we also know that
  $\varrho H$ is a quotient of $\varrho F$, therefore~$\varrho H$
  consists of those $\alpha$-equivalence classes of $\lambda$-trees
  that contain a rational $\lambda$-tree.
\end{example}
\noindent Similarly, for a binding functor $H$ arising from a binding
signature one takes its raw functor $F$. Then the rational fixpoint of
$F$ consists of all rational trees for the given binding signature, and
it follows that the rational fixpoint of $H$ consists of all rational
trees modulo $\alpha$-equivalence, i.e.~it contains precisely those
$\alpha$-equivalence classes of trees for the binding signature that have
finitely many free variables and contain a rational tree.

\subsection{Exponentiation by Orbit-Finite Strong Nominal Sets}
As in \Set, the core ingredient of functors that model various
flavours of nominal automata as coalgebras is exponentiation by the
input alphabet.  Denote by \( X^P \)
the internal hom-object witnessing the cartesian closedness of $\Nom$
(i.e.\ $(-)^P$ is right-adjoint to $(-) × P$); this nominal set
contains those maps $f: P\to X$ that are finitely supported w.r.t.~the
group action given by
\[
    (\pi\star f)(y) = \pi\cdot f(\pi^{-1}\cdot y).
\]
We will see that for $P$ orbit-finite and strong, the functor $(-)^P$
is a quotient of a polynomial functor, so that
Corollary~\ref{cor:rho-quotient} applies to $(-)^P$. In fact, we are
going to prove that
\begin{equation}\label{eq:iff}
\begin{array}{p{2.3cm}} \centering
$P$ is strong \newline
and orbit-finite
\end{array}
\iff
\begin{array}{p{3.7cm}} \centering
$(-)^P$ is a quotient of a \newline
polynomial \Nom-functor.
\end{array}
\end{equation}
It is not difficult to see that orbit-finiteness of $P$ is necessary:
\begin{proposition}
    \label{orbitFiniteNecessary}
    If $(-)^P$ is finitary, then $P$ is orbit-finite.
\end{proposition}
\begin{proof}
  Take the projection $\outr: 1×P \to P$ and consider its curried
  version $\overline \outr: 1 \to P^P$. Since $(-)^P$ is finitary and $1$
  is orbit-finite (i.e.\ finitely presentable), $\overline \outr$ factors
  through an orbit-finite subobject $j:A\hookrightarrow P$:
    \[
    \begin{tikzcd}
        1 \arrow{r}{\overline \outr}
          \arrow{dr}[below left]{\bar f}
        & P^P
        \\
        & A^P
            \arrow{u}[right]{j^P}
    \end{tikzcd}
    \]
    In other words, $\outr = j\cdot f$, where $f$ is the uncurrying of
    $\bar f$. Since $\outr$ is surjective, $j$ is surjective; hence,
    $P$ is orbit-finite because orbit-finite sets are closed under
    epimorphisms.
\end{proof}
\noindent
Secondly, we show that it is necessary that $P$ is strong. For the
sake of readability, we show the contraposition of~\eqref{eq:iff} for
a concrete example and then indicate how the construction generalizes.
\begin{proposition}\label{prop:nec}
  Let $B = \{ \{a,b\} \mid a,b \in \V, a\neq b\}$ be the (non-strong)
  nominal set of unordered pairs of distinct elements of $X$. Then the
  functor $(-)^B$ is not a natural quotient of any \Nom-functor with a
  sub-strength.
\end{proposition}
\begin{proof}
    Assume that we have a natural quotient $q_X: FX\twoheadrightarrow X^B$ and a sub-strength
    $s_{X,Y}: FX < Y \to F(X<Y)$. Then for any $Y$, the following diagram
    commutes:
    \[
    \begin{tikzcd}
        FB < Y
        \arrow{dr}[below left]{\outl}
        \arrow{r}{s_{B,Y}}
        & F(B<Y) \arrow{d}[near end]{F\outl}
        \arrow[->>]{r}{q_{B<Y}}
        \descto[pos=.1]{dl}{\text{\eqref{diag:substrength}}}
        \descto[yshift=2mm]{dr}{\text{Naturality}}
        & (B<Y)^B
        \arrow{d}{\outl^B}
        \\
        {} &
        FB
        \arrow[->>]{r}{q_{B}}
        & B^B
    \end{tikzcd}
    \]
    The identity $\id_B$ is equivariant, hence finitely supported,
    i.e.~$\id_B \in B^B$. Since $q_B$ is surjective, we have $x\in FB$
    such that $q_B(x) = \id_B$. Let $n = |\supp(x)|$ and 
    $Y= \V^n$. Then there exist $v_1,\ldots,v_n \in \V$ such that
    $(x,(v_1,\ldots,v_n)) \in FB <Y$. Put
    $g= q_{B<Y}( s_{B,Y}(x,v_1,\ldots,v_n)): B \to (B < Y)$.
    This is a finitely supported map, so we can pick distinct
    $a,b\in \V$ that are fresh for~$g$, so that $(a\,b)\star g= g$. By
    commutativity of the above diagram,
    \begin{align*}
      \outl\circ g &=\outl^B(g)=\outl^B(q_{B<Y}(
      s_{B,Y}(x,v_1,\ldots,v_n)))
      \\
      &=q_B(\outl(x,v_1,\ldots,v_n))=q_B(x)=\id_B.
    \intertext{In particular, $g(\{a, b\})$ has the form
    $g(\{a, b\})=(\{a,b\},u_1,\ldots,u_n)$ with $u_1,\dots,u_n\in\V$. Since
    $\supp(\{a,b\})\subseteq \supp(u_1,\ldots,u_n)$, we have
    $1\le i\le n$ such that $u_i = a$. Therefore,
    }
        g(\{a,b\})
        &= \big((a\,b)\star g\big)(\{a,b\})
        = (a\,b)\cdot g\big((a\,b)^{-1}\cdot \{a,b\}\big)
        = (a\,b)\cdot g\big(\{a,b\}\big)
        \\ &
        = (\{a,b\},(a\,b)\cdot u_1,\ldots,(a\,b)\cdot u_n)
        \neq g(\{a,b\}),
    \end{align*}
    in contradiction to $(a\,b) \cdot u_i = b \neq a = u_i$.
\end{proof}
\noindent The counterexample in Proposition~\ref{prop:nec} can be generalized to an arbitrary
non-strong nominal set $B$ by using, in lieu of $\{a,b\}$ and $(a\, b)$ in the
above proof, an element $z\in B$ that fails to be strongly supported
but is fresh for $g$ (i.e.\ $\supp(z)\cap\supp(g)=\emptyset$) and 
$\pi\in (\fix(g)\cap \fix(z))\setminus \Fix(\supp(z))$, respectively.
\begin{remark}
  Proposition~\ref{prop:nec} has two consequences for an
  orbit-finite non-strong nominal set $B$:
    \begin{itemize}
    \item The exponentiation functor $E = (-)^B$ is not the quotient of any
      polynomial \Nom-functor (in the sense of
      Definition~\ref{def:polynomial}), i.e.~``$\Leftarrow$'' in~\eqref{eq:iff} holds.  

    \item The exponentiation functor $E$ has no sub-strength.

    \end{itemize}
\end{remark}

\noindent A basic example of a strong nominal set is the set $P = \V$ of
all atoms. We will now show that $(-)^\V$ is a quotient of a
polynomial functor. Later, we extend this to $P = \V^n$, and then
conclude the desired result for arbitrary orbit-finite strong nominal
sets $P$.

\begin{notation}
In the following we shall write $\vec x$ for a tuple $(x_1,\ldots,x_n)\in X^n$ for any set $X$, and for a map $f: X\to Y$ we write $f(\vec x)$ for the tuple $(f(x_1), \ldots, f(x_n))$. 
\end{notation}
\noindent 
Consider the functor
\begin{equation}\label{eq:F}
  F X = \V×X×\coprod_{n\in \N} \V^n×X^n.
\end{equation}
In order to identify $(-)^\V$ as a quotient of this functor, we
define a map $\bar q_X: FX × \V \to X$:
\[
    \bar q_X(a,d,\vec v, \vec x,b) =
    \begin{cases}
        x_i & \text{where }i\text{ is minimal s.t.~}v_i = b\\
        (a\ b)\cdot d    & \text{if no such $i$ exists.}
    \end{cases}
\]
This definition of $\bar q_X$ exploits the fact that a finitely
supported map $f: \V\to X$ is equivariant w.r.t.~permutations that fix
elements in $\supp(f)$, i.e.~whenever $\pi \in \Fix(\supp(f))$ then
$f(\pi \cdot x) = \pi \cdot f(x) $ for every $x$.  (In particular, the
finitely supported maps with empty support are precisely the
equivariant maps.) Therefore, in order to represent $f$, we fix a name
$a\in\V\setminus\supp(f)$ and its image $d=f(a)$; these data then determine
the action of $f$ on all names of the form $\pi(a)$ for
$\pi\in\Fix(\supp(f))$. These are all names except those in
$\supp(f)$; we therefore enumerate the names in $\supp(f)$ as a
tuple~$\vec v$, and their images as a tuple~$\vec x$, arriving at a
representation of~$f$ as a quadruple $(a,d,\vec v,\vec x)\in FX$.
\begin{lemma}
    The map $\bar q_X: FX×\V\to X$ is equivariant and natural in $X$.
\end{lemma}
\begin{proof}
  \qquad\emph{Equivariance:} All operations used in the definition of
  $\bar q$ are equivariant, in particular the operation of picking the
  first occurrence of given name, if any, from a list of names, as
  well as the map $(a,b,d)\mapsto (a\,b)\cdot d$.

\emph{Naturality:} Let 
$f: X\to Y$ be equivariant. Then
\begin{align*}
    f\big(\bar q_X(a,d,\vec v,\vec x,b)\big)
    &=
    \begin{cases}
        f(x_i) & \text{where }i\text{ is minimal s.t.~}v_i = b\\
        f\big((a\ b)\cdot d\big)
                & \text{if no such $i$ exists.}
    \end{cases}
    \\ &=
    \begin{cases}
        f(x_i) & \text{where }i\text{ is minimal s.t.~}v_i = b\\
        (a\ b)\cdot f(d)
                & \text{if no such $i$ exists.}
    \end{cases}
    \\ &
    = \bar q_X(a,f(d),\vec v, f(\vec x), b).
    \tag*{\qedhere}
\end{align*}
\end{proof}
\noindent By currying, $\bar q$ induces a natural transformation
\begin{equation*}
q: F\to (-)^\V.
\end{equation*}
\begin{lemma}\label{lem:qSurjective}
  The natural transformation $q: F\to (-)^\V$ is component-wise
  surjective. More specifically, given $f\in X^\V$, let
  $\{v_1,\ldots,v_n\} = \supp(f)$ and $a\in \V\setminus\supp(f)$; then
  we have
    \[
         q_X(a,f(a),\vec v, f(\vec v)) = f.
    \]
\end{lemma}
\begin{proof}
  We just have to formalize the argument given in the informal
  explanation of the definition of $\bar q$: Let $b\in\V$, and put
  $\vec v = (v_1, \ldots, v_n)$,
  $g = q_X\big(a,f(a),\vec v, f(\vec v)\big): \V \to X$. We have to
  show $g(b)=f(b)$.
    \begin{itemize}
    \item If $b \in \supp(f)$, then $b = v_i$ for some $i$, so that
      $g(v_i) = f(v_i)$ by definition.
    \item If $b\in \V\setminus\supp(f)$, then $v_i \neq b$ for all
      $1\le i \le n$, so $g(b)=(a\,b)\cdot f(a)$. Moreover,
      $a,b\notin\supp(f)$ implies $(a\ b)\star f = f$. Therefore,
    \[
        g(b) = (a\ b)\cdot f(a)
             = (a\ b)\cdot f((a\ b)^{-1}\cdot (a\ b)\cdot a)
             = (a\ b) \star f((a\ b) \cdot a)
             = f((a\ b)\cdot a)
             = f(b).
    \qedhere
    \]

    \end{itemize}
\end{proof}
\noindent 
Up to now, we have seen that exponentiation by $\V$ is a quotient of a
polynomial functor, $F$~\eqref{eq:F}. To extend this to exponentiation
by $\V^n$, $n \geq 0$, recall from Lemma~\ref{lem:quot} that quotients
of polynomial $\Set$-functors compose.  Now observe that by the usual
exponentiation laws, $(-)^{\V^n}$ is just the $n$-fold composite of
$(-)^\V$ with itself. Being a polynomial functor on a cartesian closed
category, $F$ preserves epis; so $(-)^{\V^n}$ is a quotient of $F^n$
(i.e.\ of the $n$-fold composite $F\circ\dots\circ F$) by
Lemma~\ref{lem:quot}, applied inductively with trivial base case
$n=0$.

\begin{definition}
  Recall from Notation~\ref{not:neq} that $X^{n\neq}\subseteq X^n$ denotes
  the set of tuples of $n$ distinct elements, and let
  $m: X^{n\neq}\rightarrowtail X^n$ be the inclusion map. Define
    \[
    \op{uniq}: X^n \twoheadrightarrow \coprod_{1\le k\le n} X^{k\neq}
    \]
    to be the map that removes all duplicates:
    \[
        \op{uniq}(\vec x) =
        \big(
            v_i \mid 1\le i\le n, \forall j < i: v_j \neq v_i
        \big).
    \]
\end{definition}
\noindent Note that $\op{uniq}$ is equivariant (although not natural),
since $v_j\neq v_i$ iff $\pi\cdot v_j \neq \pi\cdot v_i$; moreover, we
have
    \begin{equation}
        \begin{tikzcd}
        X^{n\neq}
        \arrow[>->]{r}{m}
        \arrow[>->,shiftarr={yshift=7mm}]{rr}{\inj_n}
        &
        X^n
        \arrow[->>]{r}{\op{uniq}}
        &
        \displaystyle\coprod_{1\le k\le n} X^{k\neq}
        \end{tikzcd}
        \label{eq:uniq}
    \end{equation}
\begin{definition}
    For $n\ge 1$, define a map
    \[
        \op{fill}: X^n × X^{2n\neq} \to X^{n\neq}
    \]
    where $\op{fill}(\vec v, \vec w)$ removes duplicates from $\vec v$
    and fills the gap with components of $\vec w$ to obtain $n$
    distinct elements. Formally, we define $\op{fill}(\vec v, \vec w)$
    as the length-$n$ prefix of $\op{uniq}(\vec v) \vec w'$ where
    $\vec w' = (w_i\mid 1\le i\le 2n, w_i \not\in \vec v)$, noting
    that $\vec w'$ has at least $n$ elements. The map $\op{fill}$ is
    equivariant because $w_i\not \in \vec v$ iff
    $\pi\cdot w_i \not\in \pi\cdot \vec v$.  The diagram
    \begin{equation}
        \begin{tikzcd}[column sep = 14mm]
        X^{n\neq} × X^{2n\neq}
        \arrow[>->]{r}{m × X^{2n\neq}}
        \arrow[shiftarr={yshift=6mm}]{rr}{\outl}
        &
        X^n × X^{2n\neq}
        \arrow[->]{r}{\op{fill}}
        &
        X^{n\neq}.
        \end{tikzcd}
        \label{eq:fill}
    \end{equation}
    commutes.
  \end{definition}
  
\begin{lemma}\label{exponentUniq}
  The restriction map $r_X: X^{\V^n} \to X^{\V^{n\neq}}$ (i.e.\
  $r_X(g) = g\cdot m$) is equivariant, surjective, and natural in $X$.
\end{lemma}
\begin{proof}
  Equivariance and naturality are by standard properties of cartesian
  closed categories. We show surjectivity.
  We write $\op{eval}_{X}$ for the evaluation map
  $X^{\V^{n\neq}} × \V^{n\neq} \to X$. We then have an equivariant map
    \[
        \bar g: X^{\V^{n\neq}} × \V^{2n\neq} × \V^n\to X,
        \quad
        \bar g(f,\vec w, \vec v) =
            \op{eval}_X(f, \op{fill}(\vec v,\vec w))
    \]
    whose curried version
    $g: X^{\V^{n\neq}} \!\!\!× \V^{2n\neq} \to X^{\V^n}\!\!$ provides us
    with the desired preimage of a given $f\in X^{\V^{n\neq}}$. Indeed,
    pick any $\vec w\in \V^{2n \neq}$. Then
    $r_X(g(f,\vec w))=f$: for $\vec u\in\V^{n\neq}$, we have
    \begin{align*}
        r_X(g(f,\vec w))(\vec u)
        &
        = g(f,\vec w)(m(\vec u))
        = \bar g(f,\vec w,m(\vec u))
        = \op{eval}_X(f,\op{fill}(m(\vec u),\vec w))
        \\ &
        \overset{\mathclap{\eqref{eq:fill}}}{=}
        \op{eval}_X(f,\outl(\vec u,\vec w))
        = \op{eval}_X(f,\vec u)
        = f(\vec u).
        \tag*{\qedhere}
    \end{align*}
\end{proof}
\noindent 
This result allows us to describe the exponentiation by a nominal set
from a slightly larger class of nominal sets.
\begin{lemma}
    \label{singleOrbitExponent}
    \begin{enumerate}[(1)]
    \item\label{item:single-orbit-strong} Every single-orbit strong nominal set $P$ is isomorphic to $\V^{n\neq}$
      where $n = |\supp(p)|$, $p\in P$.
    \item \label{item:strong} Every strong nominal set is
      isomorphic to a coproduct of nominal sets of the form~$\V^{n\neq}$.
    \end{enumerate}
\end{lemma}
\begin{samepage}
\begin{proof}
\begin{enumerate}[(1)]
  \item % \emph{\ref{item:single-orbit-strong}:}
  Pick some $p\in P$ and choose
  some order $\{v_1,\ldots,v_n\} = \supp(p)$. Since
  $\supp(p)=\supp(v_1,\dots,v_n)$, the isomorphism
  $\{p\}\cong\{(v_1,\dots,v_n)\}$ induces an isomorphism
  $P\cong\V^{n\neq}$ by Proposition~\ref{orderedOrbitProperty}.

  \item % \emph{\ref{item:strong}:}
  Immediate
  from~\ref{item:single-orbit-strong}, noting that every nominal set
  is the coproduct of its orbits and orbits of strong nominal sets are
  strong.
  \qedhere
\end{enumerate}
\end{proof}
\end{samepage}
\noindent 
This combines nicely with the usual power law for coproducts:
\begin{lemma}
    \label{coproductExponent}
    Given quotients $FX \twoheadrightarrow (-)^P$, $GX
    \twoheadrightarrow (-)^Q$,
    exponentiation by $P+Q$ is a quotient of $F×G$.
\end{lemma}
\begin{proof}
  Epimorphisms are stable under products in \Nom, and
  $(-)^P\times (-)^Q\cong(-)^{P+Q}$.
\end{proof}
\noindent 
In combination, these observations prove `$\implies$' in
\eqref{eq:iff}:
\begin{corollary}
    \label{mainExponentiation}
    For any orbit-finite strong nominal set $P$, the functor $(-)^P$
    is the quotient of a polynomial functor.
\end{corollary}
\begin{proof}
  We have observed that $(-)^{\V^n}$ is a quotient of a polynomial
  functor for every~$n$. By Lemma~\ref{exponentUniq}, it follows that
  $(-)^{\V^{n\neq}}$ is a quotient of a polynomial functor. By
  Lemma~\ref{singleOrbitExponent}.\ref{item:single-orbit-strong}, this
  property extends to $(-)^P$ for every single-orbit strong nominal
  set $P$, and by Lemma~\ref{singleOrbitExponent}E.\ref{item:strong} and
  Lemma~\ref{coproductExponent} to every orbit-finite strong nominal
  set~$P$.
\end{proof}
\noindent
Putting all the previous examples together, we can sum up:
\begin{corollary}
  The class of quotients of $\Nom$-liftings contains the constant
  functors, the identity functor, $\Potf$, the abstraction functor
  $[\V]$, and the functor $(-)^P$ for any orbit-finite strong nominal
  set $P$, and is closed under coproducts, finite products,
  composition, and quotients.
\end{corollary}

\section{Conclusions and Future Work}\label{sec:conc}

We have identified a sufficient criterion for the rational fixpoint
$\varrho\bar F$ of a functor $\bar F$ on $\Nom$ that lifts a functor
$F$ on $\Set$ to arise as a lifting of the rational fixpoint
$\varrho F$ of $F$. Moreover, we have given a sufficient condition
that guarantees that rational fixpoints survive quotienting of
functors on $\Nom$, that is, for the rational fixpoint $\varrho H$ of
a quotient $H$ of a $\Nom$-functor $G$ to be a quotient of the
rational fixpoint $\varrho G$ of $G$. In combination, these results
yield a description of the rational fixpoint for quotients of liftings
of \Set-functors to $\Nom$. This applies in particular to functors
arising from combinations of binding signatures and exponentiation by
orbit-finite strong nominal sets. This includes type functors arising
in the study of nominal automata, which typically contain
exponentiation as in the functor~$2 \times X^\V \times [\V] X$
defining deterministic nominal automata~\cite{KozenEA15}.

It remains to explore the scope of these results, and possibly extend
them. Specifically, it is not currently clear how restrictive our
sufficient condition on rational fixpoints of liftings actually is; we
do give an example of a lifting that violates the condition, and for
which indeed the fixpoint of the underlying functor does not lift, but
that example is somewhat contrived and moreover can be dealt with by
moving to an isomorphic functor. Our condition on quotients of
functors in Theorem~\ref{thm:quot} makes explicit reference to coalgebras
of the quotient; the presence of a sub-strength then is a condition
that refers only to the structure of the quotiented functor as such,
without mentioning its coalgebras. We leave a closer analysis of these
conditions to future work, e.g.~the question whether there are weaker
conditions implying the condition on quotients in
Theorem~\ref{thm:quot}. 
\bibliographystyle{myabbrv}      % mathematics and physical sciences
\bibliography{refs}   % name your BibTeX data base

\begin{appendix}

\normalsize
\section*{Appendix: Finitary Functors and Preservation of Strong Epimorphisms}
\renewcommand{\thesection}{A}
We prove that for finitary functors between locally finitely
presentable categories, preservations of strong epimorphisms may be
tested on strong epimorphisms with finitely generated domain and
codomain.

Let $\C$ be a locally finitely presentable category. Recall that an
object $C$ of $\C$ is \emph{finitely generated} (fg) if its covariant
hom-functor $\C(X,-)$ preserves directed unions. Further recall that
every object of $\C$ is the directed union of all its fg subobjects
and that $\C$ has (strong epi, mono) factorizations (see
\cite[Proposition~1.61 and Theorem~1.70]{ar94}).

Note that in general the classes of finitely presentable and finitely
generated objects do not coincide. However, in the category \Nom of
nominal sets, the finitely generated objects are precisely the
orbit-finite nominal sets and the strong epimorphisms are the
surjective equivariant maps (i.e.~all epis are strong).

\begin{lemma} \label{colimitepi}
    For any directed diagram $D: (I,\le) \to \C$ of subobjects $m_i: C_i\rightarrowtail
    C$ of $C$, the colimit $(d_i: C_i\to \colim D)_{i\in I}$ is obtained by
    taking the (strong epi,mono)-factorization of $\coprod C_i
    \xrightarrow{[m_i]} C$.
\end{lemma}
\begin{proof}
  First note that the $(m_i)_{i\in \D}$ form a cocone, so we have a unique
  $m: \colim D \to C$ with ${m\cdot d_i = m_i}$, and $d_i$ is
  monic. As $\C$ is lfp and both $d_i$ and $m_i$ are monic,
  \cite[Proposition~1.62(ii)]{ar94} implies that $m$ is monic,
  too. Recall that, in general, the copairing of colimit injections
  yields a strong epimorphism $[d_i]: \coprod C_i\to \colim D$. Therefore
  we have the factorization:
    \[
    \begin{tikzcd}%[baseline=(bottom.base)]
        \coprod C_i \arrow{rr}{[m_i]}
        \arrow[->>]{dr}[below left]{[d_i]}
        & & C
        \\
        & % |[alias=bottom]|
        \colim D
        \arrow[>->]{ur}[below right]{m}
    \end{tikzcd}
        \tag*{\qedhere}
    \]
\end{proof}

\begin{lemma}\label{unionsofimages}
  Strong quotients of directed colimits are directed colimits of images. More
  precisely, for a diagram $D: \D\to \C$, given a colimit cocone
  $(c_i: Di \to C)_{i\in \D}$ and a strong epimorphism
  $e: C\twoheadrightarrow B$, define $A_i$ by factorizing $e\cdot c_i$
  into a strong epi and a mono. Then $B$ is the directed colimit of
  the $A_i$ together with the induced monomorphisms.
\end{lemma}
\begin{proof}
    For each $i\in \D$, take the (strong epi,mono)-factorization
    \[
    \begin{tikzcd}
        Di
            \arrow[shiftarr={yshift=5mm}]{rr}{e\cdot c_i}
            \arrow[->>]{r}{e_i}
        & A_i
            \arrow[>->]{r}{m_i}
        & B.
    \end{tikzcd}
    \]
    For any morphism $g: Di\to Dj$ we get a morphism $\bar g: A_i\rightarrowtail A_j$ by diagonalization:
    \[
        \begin{tikzcd}
            Di\arrow[->>]{r}{e_i}
            \arrow{d}[left]{g}
            & A_i \arrow[>->]{dr}[above right]{m_i}
            \arrow[dashed,>->]{d}[left]{\bar g}
            \\
            Dj\arrow[->>]{r}{e_j}
            &
            A_j \arrow[>->]{r}[below]{m_j}
            &
            B
        \end{tikzcd}
    \]
    Since $d_j\cdot\bar g = d_i$, we see that $\bar g$ is monic. It is
    easy to see that the $A_i$ form a directed diagram of monos in
    $\C$.  To see that $B$ is indeed its colimit, consider the square
    \[
        \begin{tikzcd}
            \coprod_{i} Di
            \arrow[->>]{r}{[c_i]}
            \arrow[->>]{d}[left]{\coprod e_i}
            &
            C
            \arrow[->>]{d}{e}
            \\
            \coprod A_i
            \arrow{r}[below]{[m_i]}
            &
            B
        \end{tikzcd}
    \]
    which commutes by the definition of $e_i$ and $m_i$. The copairing
    of the colimit injections $[c_i]$ is a strong epi, hence so is
    $e\cdot [c_i]$. Since $\coprod e_i$ is a strong epi as well, we
    see that $[m_i]$ is a strong epi. By Lemma~\ref{colimitepi}, it
    follows that $B$ is the colimit of the $A_i$ as desired.
\end{proof}

\begin{proposition}
    \label{finitaryEpiPreservation}
    Let $\C$ and $\D$ be locally finitely presentable categories, and let 
    $F: \C\to\D$ be a finitary functor preserving strong epimorphisms
    with finitely generated domain and codomain. Then $F$
    preserves all epimorphisms.
\end{proposition}
\begin{proof}
  Let $e: X\to Y$ be a strong epimorphism.  Write $X$ as the colimit
  of the directed diagram of all its finitely generated subobjects
  $c_i: X_i\to X$. Take the (strong epi, mono)-factorizations of all
  $e_i \dot c$:
    \[
    \begin{tikzcd}
        X_i
            \arrow[shiftarr={yshift=5mm}]{rr}{e\cdot c_i}
            \arrow[->>]{r}{e_i}
        & A_i
            \arrow[>->]{r}{m_i}
        & Y.
    \end{tikzcd}
    \]
    Note that each $A_i$ is finitely generated, being a strong
    quotient of the finitely generated object~$X_i$. By
    Lemma~\ref{unionsofimages}, $Y$ is the directed colimit of the $A_i$
    with colimit injections $m_i$. This directed colimit is preserved
    by the finitary functor $F$, resulting in a colimit cocone
    $(Fm_i: FA_i \to FY)$. The family of colimit injections $F m_i$ is
    jointly strongly epic. By assumption, each of the strong
    epimorphisms $e_i$ is preserved by $F$. Hence the
    $F m_i\cdot Fe_i$ form a jointly strongly epic family. Since the colimit
    injections $Fc_i$ form jointly strongly epic family and 
    \[
        Fe\cdot Fc_i = Fm_i\cdot Fe_i,
    \]
    we conclude that $Fe$ is a strong epimorphism.
\end{proof}
\end{appendix}

\end{document}